%% file: main.tex
\newcommand\arxiv

\ifdefined\tacas
\documentclass[runningheads]{llncs}
\usepackage{times}
\else
\documentclass{article}
\usepackage{authblk}
\usepackage{fullpage}
\fi

\usepackage{amsmath}

\usepackage{amsthm}
\newtheorem{problem}{Problem}
\newtheorem{theorem}{Theorem}[section]
\newtheorem{corollary}[theorem]{Corollary}
\newtheorem{lemma}[theorem]{Lemma}

\usepackage{url}
\usepackage{amssymb}
\usepackage{zed-csp}
\usepackage{graphicx}
\usepackage{xspace}
\usepackage{xcolor}
\usepackage{mdframed}
\newcommand{\para}[1]{\paragraph{\textbf{#1.}}}

\usepackage[normalem]{ulem}  % for \sout{ .. } command
\usepackage{wrapfig}

\usepackage{etoolbox}
\patchcmd{\quote}{\rightmargin}{\leftmargin 1em \rightmargin}{}{}

\usepackage[linesnumbered]{algorithm2e}

\newcommand{\red}[1]{{{\textcolor{red}{}}}}

\newcommand{\tuple}[1]{\ensuremath{\left \langle #1 \right \rangle }}

\renewcommand\bf[1]{\textbf{#1}}
\renewcommand\sf[1]{\textsf{\small{#1}}}

\ifdefined\tacas
\newenvironment{exampleregion}{}{}
\usepackage{comment}
\excludecomment{exampleregion}
\else 
\newenvironment{exampleregion}
{\mdfsetup{
hidealllines=true,
innerleftmargin=3pt,
innerrightmargin=3pt,
leftmargin=-3pt,
rightmargin=-3pt,
nobreak=false,
backgroundcolor=lightgray!20
}
\begin{mdframed}%
}
{\end{mdframed}}
\fi

\input{macros}

\ifdefined\techreport
\begin{document}
\else
% Copyright
%\setcopyright{none}
%\setcopyright{acmcopyright}
%\setcopyright{acmlicensed}
%\setcopyright{rightsretained}
%\setcopyright{usgov}
%\setcopyright{usgovmixed}
%\setcopyright{cagov}
%\setcopyright{cagovmixed}

% DOI
%\acmDOI{10.475/123_4}

% ISBN
%\acmISBN{123-4567-24-567/08/06}

%Conference
%\acmConference[WOODSTOCK'97]{ACM Woodstock conference}{July 1997}{El
%  Paso, Texas USA} 
%\acmYear{1997}
%\copyrightyear{2016}

%\acmPrice{15.00}

\begin{document}

\makeatletter
\renewcommand{\paragraph}{%
  \@startsection{paragraph}{4}%
  {\z@}{1.25ex \@plus 1ex \@minus .2ex}{-1em}%
  {\normalfont\normalsize\bfseries}%
}
\makeatother

\fi

\ifdefined\tacas
\title{Automated Synthesis of Secure Platform Mappings} 

\author{Eunsuk Kang\inst{1} \and
 St\'{e}phane Lafortune\inst{2} \and
 Stavros Tripakis\inst{3} 
}

\institute{Carnegie Mellon University \texttt{eskang@cmu.edu} \and
University of Michigan \texttt{stephane@umich.edu} \and
Northeastern University \texttt{stavros@northeastern.edu}
}
\else
\title{Synthesis of Property-Preserving Mappings}

\author{Eunsuk Kang\thanks{School of Computer Science. e-mail: eskang@cmu.edu}}
\affil{Carnegie Mellon University}

\author{St\'{e}phane Lafortune\thanks{Department of Electrical
    Engineering and Computer Science. e-mail: stephane@umich.edu}}
\affil{University of Michigan, Ann Arbor}

\author{Stavros Tripakis\thanks{College of Computer and Information
    Science. e-mail: stavros@northeastern.edu}}
\affil{Northeastern University}

\fi
% \titlenote{Produces the permission block, and
% \subtitle{Extended Abstract}
% \subtitlenote{The full version of the author's guide is available as
%   \texttt{acmart.pdf} document}

\maketitle

\begin{abstract}
  System development often involves decisions about how a high-level
  design is to be implemented using primitives from a low-level
  platform. Certain decisions, however, may introduce undesirable
  behavior into the resulting implementation, possibly leading to a
  violation of a desired property that has already been established at
  the design level. In this paper, we introduce the problem of
  \textit{synthesizing a property-preserving platform mapping}: A set
  of implementation decisions ensuring that a desired property is
  preserved from a high-level design into a low-level platform
  implementation. We provide a formalization of the synthesis problem
  and propose a technique for synthesizing a mapping based on symbolic
  constraint search. We describe our prototype implementation, and a
  real-world case study demonstrating the application of our technique to
  synthesizing secure mappings for the popular web authorization
  protocols OAuth 1.0 and 2.0.
\end{abstract}



% We no longer use \terms command

\ifdefined\techreport
\else
%\keywords{Modeling, composition, synthesis, verification, security.}
\fi

%% somehow the ACM style seems to need the abstract before the \maketitle ...
\ifdefined\techreport
\red{Stavros main comments, Aug 30, 2017:
(1) notation: the paper is sloppy with notation at several places, using different notations for the same thing, or the same notation for different things;
(2) possible duplication in semantics equivalence proof: see comments end of Section~3;
(3) Section~4 out of focus: see comments there;
(4) Section~5: perhaps the discussion on reusing the same specification for high and low-level verification can be simplified? please see comments in that section;
(5) I don't think the dataflow extension should be in the appendix. see my comment there.
}
\fi

\ifdefined\techreport
\input{introduction}
\input{example}
\input{mapping}

\input{modeling-interaction}
\input{mapping-verification}
\input{mapping-synthesis}
\input{synthesis-technique}

\input{dataflow}
\input{case-study}

\input{related-work}
\input{conclusion}
\else

\input{introduction}
\input{example}
\input{mapping}
\input{mapping-verification}
\input{mapping-synthesis}
\input{synthesis-technique}
\input{case-study}
\input{related-work}
\input{conclusion}
\fi

\ifdefined\arxiv
\section*{Acknowledgements}

This work was partially supported by the NSF SaTC award ``Bridging the
Gap between Protocol Design and Implementation through Automated
Mapping'' (CNS-1801546),
the NSF Breakthrough award ``Compositional System Modeling''
(CNS-132975), the NSF Expeditions in Computing project ``ExCAPE:
Expeditions in Computer Augmented Program Engineering'' (CCF-1138860
and 1139138), and by the Academy of Finland.

\appendix

\input{appendix-proofs}
\fi

\ifdefined\techreport
\bibliographystyle{plain}
\else
\bibliographystyle{plain}
\fi
\bibliography{sigproc} 

\ifdefined\arxiv
\else
\para{Acknowledgements}
This work was partially supported by the NSF SaTC award CNS-1801546.
\fi
\end{document}

%% file: macros.tex
\newcommand{\e}[1]{\{#1\}}
\newcommand{\traces}{\mathit{beh}}
\newcommand{\events}{\mathit{evts}}
\newcommand{\labels}{\alpha}
\newcommand{\compo}{\|_m}

\newcommand{\mparw}[3]{#1 \|_{#3} #2}

\newcommand{\Alice}{\sf{Alice}\xspace}
\newcommand{\Bob}{\sf{Bob}\xspace}
\newcommand{\secret}{{\sf{secret}\xspace}}
\newcommand{\none}{{\sf{none}\xspace}}
\newcommand{\public}{{\sf{public}\xspace}}
\newcommand{\keyX}{\sf{k}_\sf{{X}}}
\newcommand{\keyY}{\sf{k}_\sf{{Y}}}
\newcommand{\writeBob}[1]{{\sf{writeBob}(#1)\xspace}}
\newcommand{\writeEve}[1]{{\sf{writeEve}(#1)\xspace}}
\newcommand{\knows}{{\sf{knows}}}

\newcommand{\Msg}{\sf{Msg}\xspace}

\newcommand{\Sender}{\sf{Sender}\xspace}
\newcommand{\Receiver}{\sf{Receiver}\xspace}
\newcommand{\ReceiverX}{\sf{Receiver}_\sf{X}\xspace}
\newcommand{\ReceiverY}{\sf{Receiver}_\sf{Y}\xspace}

\newcommand{\Key}{\sf{Key}\xspace}

\newcommand{\encWrite}[2]{\sf{encWrite}(#1,#2)\xspace}

\newcommand{\req}{\sf{req}\xspace}
\newcommand{\URLL}{\sf{URL}\xspace}
\newcommand{\Host}{\sf{Host}\xspace}
\newcommand{\Path}{\sf{Path}\xspace}
\newcommand{\Query}{\sf{Query}\xspace}
\newcommand{\Body}{\sf{Body}\xspace}
\newcommand{\Name}{\sf{Name}\xspace}
\newcommand{\Value}{\sf{Value}\xspace}
\newcommand{\Status}{\sf{Status}\xspace}
\newcommand{\Method}{\sf{Method}\xspace}
\newcommand{\Header}{\sf{Header}\xspace}
\newcommand{\header}{\sf{header}\xspace}

\newcommand{\Response}{\sf{Resp}\xspace}
\newcommand{\resp}{\sf{resp}\xspace}
\newcommand{\List}{\sf{List}\xspace}

\newcommand{\Eve}{\sf{Eve}\xspace}
\newcommand{\Client}{\sf{Client}\xspace}
\newcommand{\AuthServer}{\sf{AuthServer}\xspace}

\newcommand{\initiate}{\sf{initiate}\xspace}
\newcommand{\authorize}{\sf{authorize}\xspace}
\newcommand{\forward}{\sf{forward}\xspace}
\newcommand{\getToken}{\sf{getToken}\xspace}

\newcommand{\getReqToken}{\sf{getReqToken}\xspace}
\newcommand{\notify}{\sf{notify}\xspace}
\newcommand{\getAccessToken}{\sf{getAccessToken}\xspace}

\newcommand{\Browser}{\sf{Browser}\xspace}
\newcommand{\Server}{\sf{Server}\xspace}

%% file: introduction.tex
\section{Introduction}
\label{sec-introduction}

When building a complex software system, one begins by
coming up with an abstract design, and then constructs an implementation that
conforms to this design. In practice, there are rarely enough time and
resources available to build an implementation from scratch, and so
this process often involves reuse of an existing \emph{platform}---a
collection of generic components, data structures, and libraries that
are used to build an application in a particular domain.

The benefits of reuse also come with potential risks. A typical
platform exhibits its own complex behavior, including subtle
interactions with the environment that may be difficult to anticipate
and reason about. Typically, the developer must work with the platform
as it exists, and is rarely given the luxury of being able to modify
it and remove unwanted features. For example, when building a web
application, a developer must work with a standard browser and take
into account all its features and security vulnerabilities. As a
result, achieving an implementation that perfectly conforms to the
design---in the traditional notion of behavioral
refinement~\cite{hoare-refinement}---may be too difficult in
practice. Worse, the resulting implementation may not necessarily
preserve desirable properties that have already been established at
the level of design.

These risks are especially evident in applications where security is a
major concern. For example, OAuth 2.0, a popular authorization
protocol subjected to rigorous and formal analysis at an abstract
level~\cite{oauth-verification1,oauth-verification2,oauth-verification3},
has been shown to be vulnerable to attacks when implemented on a web
browser or a mobile
device~\cite{oauth-study-sun,oauth-study-wang,oauth-study-chen}. Many
of these vulnerabilities are not due to simple programming errors:
They arise from logical flaws that involve a subtle interaction
between the protocol logic and the details of the
underlying platform. Unfortunately, OAuth itself does not explicitly
guard against these flaws, since it is intended to be a
\textit{generic}, \textit{abstract} protocol that deliberately omits
details about potential platforms. On the other hand,
anticipating and mitigating against these risks require an in-depth
understanding of the platform and security expertise, which many
developers do not possess.

% Currently, both designers and developers seem to approach this problem
% in an ad hoc, trial-and-error manner. Often it takes a high-profile
% disclosure on an existing implementation for the protocol designer to
% devise guidelines and advisories for mitigating specific
% vulnerabilities [70, 71, 72, 81]. Some vulnerabilities are not
% discovered until long after the protocol is released, during which
% critical resources that the protocol purports to protect may have
% already been exposed to malicious agents.

This paper proposes an approach to help developers overcome these
risks and achieve an implementation that preserves desired
properties. In particular, we formulate this task as the problem of
automatically synthesizing a \textit{property-preserving
 platform mapping}: A set of implementation decisions ensuring
that a desired property is preserved from a high-level design into a
low-level platform implementation. 

Our approach builds on the prior work of Kang et
al.~\cite{eunsuk-security}, which proposes a modeling and verification
framework for reasoning about security attacks across multiple levels
of abstraction. The central notion in this framework is that of a
\textit{mapping}, which captures a developer's decisions about how
abstract system entities are to be realized in terms of their concrete
counterparts. In this paper, we extend the framework with the novel problem of
synthesizing a property-preserving mapping, and propose an algorithmic
technique for performing this synthesis task.

We have built a prototype implementation of the synthesis
technique. Our tool accepts a high-level design model, a desired
system property (both specified by the developer), and a model of a
low-level platform (built and maintained separately by a domain
expert). The tool then produces a mapping (if any) that ensures that
the resulting platform implementation preserves the given property. As
a case study, we have successfully applied our tool to synthesize
property-preserving mappings for two different variants of OAuth
implemented on top of HTTP. Our results are promising: The
implementation decisions captured by our synthesized mappings describe
effective mitigations against some of the common vulnerabilities that
have been found in deployed OAuth
implementations~\cite{oauth-study-sun,oauth-study-wang}. The
contributions of this paper include:
\begin{itemize}
\item A formal treatment of mapping, including a correction in the original
  definition~\cite{eunsuk-security} (Section~\ref{sec-formal-model});
    \ifdefined\techreport
\item A novel operational semantics for mapping, which allows more intuitive
  modeling using available tools, and a proof that the operational semantics
  is equivalent to the denotational semantics
  (Section~\ref{sec-operational-semantics});
  \fi
\item A formulation of the \textit{mapping synthesis problem}, a novel
  approach for ensuring the preservation of a property between a
  high-level design and its platform implementation
  (Section~\ref{sec-synthesis-problems});
\item A technique for automatically synthesizing mappings based on
  symbolic constraint search (Section~\ref{sec-synthesis-technique});
  \ifdefined\techreport
\item Discussions of other features of the mapping-based modeling
  approach, including mappings as a communication mechanism
  (Section~\ref{sec-mapping-communication}) and dataflow modeling
  (Section~\ref{sec-dataflow}); and
  \fi
\item A prototype implementation of the synthesis technique, and
  a real-world case study demonstrating the feasibility of this approach
  (Section~\ref{sec-case-study}).
\end{itemize}
We conclude with a discussion of related work
(Section~\ref{sec-related-work}).

%% file: example.tex
\section{Running Example}
\label{sec-example}

\begin{figure*}[!t]
\centering
\includegraphics[width=0.65\textwidth]{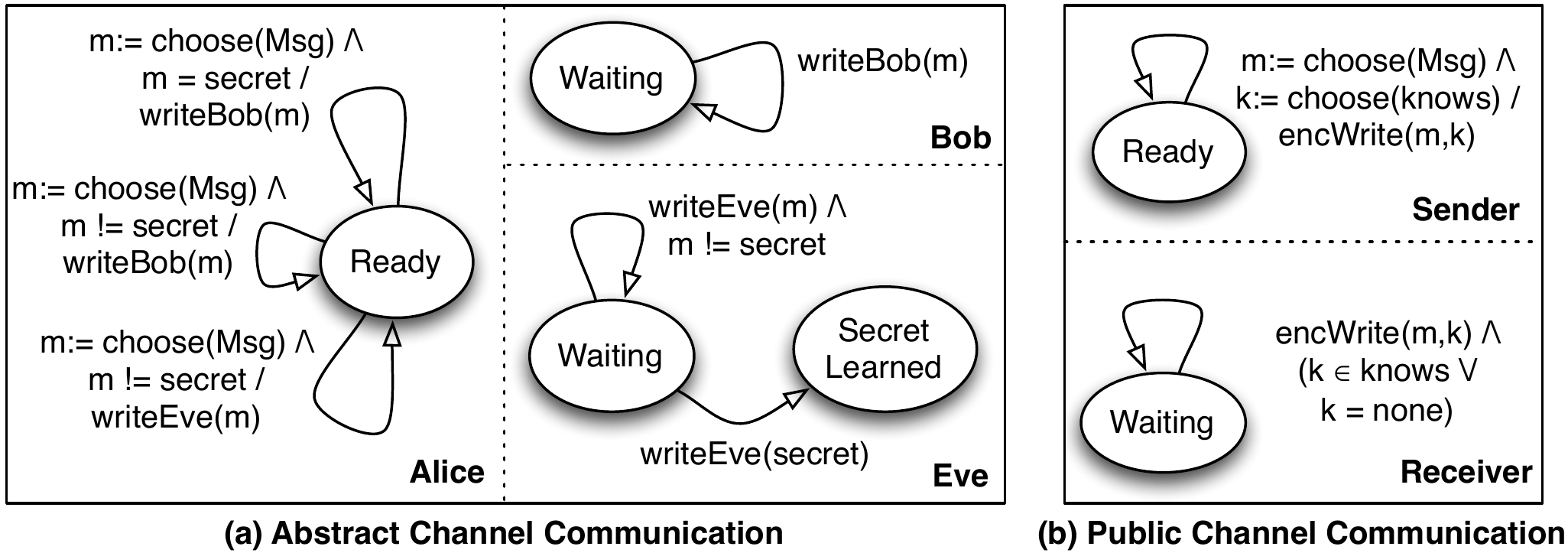}
\caption{\small{State machines describing a pair of high-level and
    low-level communication models. Each transition is in form
    $cond/out$, where $cond$ is the triggering condition, and $out$ is
    the label describing an output event; $out$ is
    optional. $choose(X)$ is a built-in function that
    non-deterministically returns a member of set $X$. For example, in
    the \Eve machine, the self-looping transition from state
    \sf{Waiting} takes place only if it receives a \sf{writeEve}
    event with some message $m$ that is not a secret.}}
% \red{i suggest to replace ";" with $\land$ in the conditions, so that for example, m:=choose(Msg); m=secret" becomes m:=choose(Msg)$\land$ m=secret.  Even better, factor out the choose command, since it's not really part of the condition. You can use something like the "junctions" (intermediate states) of Simulink: first the choose command takes you to one of those states. Then, depending on whether m is secret or not, you go to the corresponding "real" state.  You can use $\bullet$ to denote those intermediate states. Eunsuk: I am reluctant to introduce a new notation (junctions) unless absolutely necessary, so I will keep it as it is.}
\label{fig-example}
\end{figure*}

Let us introduce a running example that we will use throughout the
paper. Consider the simple model in Figure~\ref{fig-example}(a),
consisting of three interacting processes, \Alice, \Bob, and \Eve.  In
this system, \Alice wishes to communicate messages to \Bob and \Eve,
but is willing to share its secrets only with \Bob. \Alice has access
to two separate communication channels (represented by event labels
\sf{writeBob} and \sf{writeEve}). \Alice behaves as follows: It first
non-deterministically selects a message to be sent ($m$) from some set
$\sf{Msg}$. If the selected message is not a secret (represented by
the constant $\sf{secret} \in \Msg$), \Alice sends $m$ to either \Bob
or \Eve over the corresponding channel (i.e., by performing
\sf{writeBob} or \sf{writeEve}). If, however, the chosen message $m$ is a
secret, \Alice sends $m$ only to \Bob.

Suppose that \Eve is a malicious character whose goal is to learn the
secret shared between \Alice and \Bob. One desirable property of the
system is that \Eve should never be able to learn the secret. This property can be stated as 
the following first-order logic formula
\begin{align*}
\forall t \in T : \neg \exists e \in
  events(t) : \sf{writeEve}(\secret)  \in e
\end{align*}
where $T$ is the set of possible system traces, and $events(t)$
returns the set of events in trace $t$\footnote{The behavioral
  semantics will be further explained in
  Section~\ref{sec-formal-model}.}. In other words, the property
states that the system should never enter the $\sf{SecretLearned}$
state by executing a $\sf{writeEve}(\secret)$ event
(see Figure~\ref{fig-example}). It can be
observed that the composition of the three processes,
$\Alice \parallel \Bob \parallel \Eve$, satisfies this property, since \Alice, by
design, never sends a secret over the channel to \Eve.

% Suppose that \Eve is a malicious character whose goal is to learn the
% secret shared between \Alice and \Bob. One desirable property of the
% system is that \Eve should never be able to learn the secret. In
% linear temporal logic (LTL)~\cite{pnueli-ltl}, this property can be stated as 
% formula $\neg\lozenge(\sf{SecretLearned})$, which states that the
% system should never enter the state $\sf{SecretLearned}$ in every
% possible execution.  \red{we haven't introduced LTL and depending on
%   the audience we may need to say a bit more. also, there is a
%   deadlock in the model and this might create pedantic problems in the
%   satisfaction/violation of the LTL formula: finite traces (those that
%   deadlock) do not count for satisfaction/violation of LTL.}  It can
% be observed that the composition of the three processes,
% $\Alice \parallel \Bob \parallel \Eve$, satisfies this property, since \Alice, by
% design, never sends a secret over the channel to \Eve.

Now consider the model in Figure~\ref{fig-example}(b), which describes
communication between a pair of processes over an encrypted public
channel (represented by event label $\encWrite{m}{k}$, where $m$ is
the message being sent, and $k$ is the key used to encrypt the
message). Each process is associated with a value called \sf{knows},
which represents the set of keys that the process has access
to\ifdefined\techreport\footnote{The concept of \textit{knows}, along with our approach to
  dataflow modeling, is described in more detail in
  Section~\ref{sec-dataflow}.}\fi.  In its behavior, \Sender
non-deterministically chooses a message and a key to encrypt it
with. \Receiver, upon receiving the message, is able to read it only
if it knows the key $k$ that $m$ is encrypted with, or if $m$ has
not been encrypted at all (i.e., $k = \sf{none}$). In this particular
example, we use the absence of a corresponding transition to model the
assumption that \Receiver is unable to read message $m$ without
knowing the key.

Suppose that we wish to reason about the behavior of the abstract
communication system from Figure~\ref{fig-example}(a) when it is
implemented over the public channel in
\ref{fig-example}(b). Conceptually, implementation 
involves {\em mapping} the elements of the abstract, high-level model
into those of the low-level model.  The developer's task is to
determine such a mapping, and by doing so resolve various
implementation decisions that need to be made.  In this toy example,
the crucial implementation decision stems from the fact that, while
the abstract model provides two separate channels, \sf{writeBob} and
\sf{writeEve}, the low-level model provides only one channel,
\sf{encWrite}.  Therefore, the developer must decide how the abstract
events \sf{writeBob(m)} and \sf{writeEve(m)} are to be mapped into
\sf{encWrite(m,k)} events. In particular, to ensure the security of
the resulting implementation, the developer must answer: Which of the
decisions will ensure that secret messages remain protected from \Eve
in the resulting implementation? In the rest of this paper, we
describe how this question can be formulated and tackled as the
problem of synthesizing a \textit{property-preserving mapping} between
a pair of models that depict a high-level design and a low-level
platform.

%% file: mapping.tex
\section{Mapping Composition}
\label{sec-formal-model}

Our synthesis approach builds on the modeling and verification
framework proposed in \cite{eunsuk-security}, which is designed to
allow reasoning about behavior of processes across multiple
abstraction layers. In this framework, a trace-based semantic model
(based on CSP~\cite{hoare-csp}) is extended to represent events as
\textit{sets of labels}, and includes a new composition operator based
on the notion of \textit{mappings}, which relate event labels from one
abstraction layer to another. In this section, we present the
essential elements of this framework. \ifdefined\arix We also correct a flaw in the
previous (denotational) definition of mapping composition proposed
in~\cite{eunsuk-security}.  \fi

\ifdefined\arxiv
\subsection{Denotational Semantics}
\label{sec-modeling-framework}
\fi

\ifdefined\tacas
\begin{figure*}[!t]
\centering
\includegraphics[width=0.75\textwidth]{diagrams/public-channel-example.pdf}
\vspace{-2.5mm}
\caption{\small{A pair of high-level and
    low-level communication models. Each transition is in form
    $cond/out$, where $cond$ is the triggering condition, and $out$ is
    an optional label describing an output event. $choose(X)$ is a function that
    non-deterministically selects a member of set $X$.}}
% \red{i suggest to replace ";" with $\land$ in the conditions, so that for example, m:=choose(Msg); m=secret" becomes m:=choose(Msg)$\land$ m=secret.  Even better, factor out the choose command, since it's not really part of the condition. You can use something like the "junctions" (intermediate states) of Simulink: first the choose command takes you to one of those states. Then, depending on whether m is secret or not, you go to the corresponding "real" state.  You can use $\bullet$ to denote those intermediate states. Eunsuk: I am reluctant to introduce a new notation (junctions) unless absolutely necessary, so I will keep it as it is.}
\label{fig-example}
\vspace{-5mm}
\end{figure*}

\para{Running example} Consider the simple model in
Figure~\ref{fig-example}(a), consisting of three interacting
processes, \Alice, \Bob, and \Eve.  In this system, \Alice wishes to
communicate messages to \Bob and \Eve, but is willing to share its
secrets only with \Bob. \Alice has access to two separate
communication channels (represented by event labels \sf{writeBob} and
\sf{writeEve}). \Alice non-deterministically selects a message
to be sent ($m$) from some set $\sf{Msg}$. If the selected message is
not a secret (represented by the constant $\sf{secret} \in \Msg$),
\Alice sends $m$ to either \Bob or \Eve over the corresponding channel
(i.e., by performing \sf{writeBob} or \sf{writeEve}). If the
chosen message $m$ is a secret, \Alice sends $m$ only to \Bob.

\Eve is a malicious character whose goal is to learn the
secret shared between \Alice and \Bob. A desirable property of the
system is that \Eve should never be able to learn the secret, which
can be stated as the LTL formula
$\neg\lozenge(\sf{SecretLearned})$. It can be observed that the
composition of the three processes, $\Alice \| \Bob \| \Eve$,
satisfies the property, since \Alice, by design, never sends a secret
over the channel to \Eve.

The model in Figure~\ref{fig-example}(b) describes communication
between a pair of processes over an encrypted public channel
(represented by event label $\encWrite{m}{k}$, where $m$ is the
message being sent, and $k$ is the key used to encrypt the
message). Each process is associated with a value called \sf{knows},
which represents the set of keys that the process has access to.
\Sender non-deterministically chooses a message and a key to encrypt
it with. \Receiver, upon receiving the message, is able to read it
only if it knows the key $k$ that $m$ is encrypted with, or if $m$ has
not been encrypted at all (i.e., $k = \sf{none}$). 

Suppose that we wish to reason about the behavior of the abstract
communication system from Figure~\ref{fig-example}(a) when it is
implemented over the public channel in \ref{fig-example}(b). In
particular, in the low-level implementation, \Bob and \Eve are
required to share the same channel (\sf{enc}), no longer benefitting
from the separation provided by the abstraction in
Figure~\ref{fig-example}(a). \ifdefined\techreport The developer's task is to decide how the
abstract events \sf{sendBob} and \sf{sendEve} are to be implemented as
\sf{enc}.\fi Does the property of the abstract communication hold in
every possible implementation? If not, which decisions will ensure
that secret messages remain protected from \Eve? We formulate these questions as the problem of
synthesizing a \textit{property-preserving mapping} between a pair of
high-level and low-level models.  \fi

%\todo{Eunsuk: Emphasize benefits of this modeling style}

\para{Events, traces, and processes} Let $L$ be a potentially
infinite set of labels. An {\em event} $e$ is a finite, non-empty set
of labels: $e \in E(L) = \mathbb{P}(L)-\{\emptyset\}$, where
$\mathbb{P}(L)$ is the powerset of $L$, and $\emptyset$ denotes the
empty set. Let $S^*$ denote the set of all finite sequences of
elements of some set $S$. A {\em trace} $t$ is a finite sequence of
events: $t\in T(L)$, where $T(L)$ is the set of all traces over $L$
(i.e.,\ $T(L) = (E(L))^*$).  The empty trace is denoted by $\tuple{}$, and the
trace consisting of a sequence of events $e_1,e_2,...$ is denoted
$\tuple{e_1,e_2,...}$.  If $t$ and $t'$ are traces, then $t\cdot t'$
is the trace obtained by concatenating $t$ and $t'$. Note that
$\tuple{}\cdot t =t\cdot\tuple{} = t$ for any trace $t$. The {\em
  events of a trace $t$}, written $events(t)$, is the set of all
events appearing in $t$. For example, if
$t = \tuple{\{a\},\{a, c\}, \{b\}}$, then
$events(t)=\{ \{a\}, \{a, c\}, \{b\}\}$.

\ifdefined\techrep The set of all events over label set $L$ is
$E(L) = \mathbb{P}(L)$. An event may be an empty set, denoted $\emptyset$.
$$
events(t)= \{ e \subseteq L \mid t = t_1 \cdot \tuple{e} \cdot t_2,
	\mbox{ for some traces } t_1,t_2 \}
$$
\fi

Let $t$ be a trace over set of labels $L$, and let $A \subseteq L$ be
a subset of $L$.  The {\em projection of $t$ onto $A$},
denoted $t \project A$, is defined as follows:
\begin{align*}
\tuple{} \project A = \tuple{} \qquad (\tuple{e}\cdot t) \project A = \left\{\begin{array}{ll}
\tuple{e \cap A} \cdot (t \project A) & \text{if}\ e \cap A \neq \emptyset \\
(t \project A) & \text{otherwise}
\end{array}\right.
\end{align*}
For example, if $t=\tuple{\{a\},\{a, c\},\{b\}}$, then
$t \project\{a,b\} = \tuple{\{a\},\{a\},\{b\}}$ and 
$t \project\{b,c\} = \tuple{\{c\},\{b\}}$.

A {\em process} $P$ is defined as a triple
$(\labels,\events,\traces)$. The {\em alphabet} $\labels
\subseteq L$ is the set of all labels appearing in
$P$, and $\events \subseteq
E(\labels)$ is the set of events that \emph{may} appear in traces of
$P$, which are denoted $\traces \subseteq
T(\labels)$.  We assume traces in every process
$P$ to be \textit{prefix-closed}; i.e.,\ $\tuple{} \in
\traces$ and for every non-empty trace $t' = t \cdot \tuple{e} \in
\traces$, $t \in \traces$.

\para{Parallel composition}
A pair of processes $P$ and $Q$ synchronize with each other by
performing events $e_1$ and $e_2$, respectively, if these two events share at
least one label. In their parallel composition, denoted $P  \parallel Q$, this
synchronization is represented by a new event $e'$ that is constructed
as the union of $e_1$ and $e_2$ (i.e.,\ $e' = e_1 \cup
e_2$). 

Formally, let $P = (\alpha_P, \events_P, \traces_P)$ and $Q =
(\alpha_Q, \events_Q, \traces_Q)$ be a pair of proceses. Their
parallel composition is defined as follows:
\begin{align*}
& \events_{P \parallel Q} = \{ e \in E(\alpha_P \cup \alpha_Q) \mid
  e \cap \alpha_P \in \events_P \land e \cap \alpha_Q \in \events_Q
  \land cond(e) \} \\
 & \traces_{P \parallel Q} = \{ t \in (\events_{P \parallel Q})^* \mid 
(t \project \alpha_P) \in \traces_P \land (t \project
   \alpha_Q) \in \traces_Q \} \tag{\bf{Def.\ 1}}
\end{align*}
where $\alpha_{P \parallel Q} = \alpha_P \cup \alpha_Q$, and $cond(e)$ is defined as
\begin{align*}
cond(e) \equiv e \subseteq \alpha_P - \alpha_Q \lor e \subseteq
  \alpha_Q - \alpha_P \lor (\exists a \in e : a \in \alpha_P \cap
  \alpha_Q) \tag{\bf{Cond.\ 1}}
\end{align*}
The definition of $\traces_{P \parallel Q}$ states that if we take a
trace $t$ in the composite process and ignore labels that appear only
in $Q$, then the resulting trace must be a valid trace of $P$ (and
symmetrically for $Q$). The condition (\bf{Cond.\ 1}) is imposed on
every event appearing in $traces_{P \parallel Q}$ to ensure that an
event performed together by $P$ and $Q$ contains at least one common
label shared by both processes.

% Let $P_i = (L_i,E_i,T_i)$, for $i=1,2$. $P_i$ is the process. $L_i\subseteq L$ is the set of labels of $P_i$. $E_i \subseteq E(L_i)$ is the set of events of $P_i$. And $T_i \subseteq T(L_i)$ is the set of traces of $P_i$. $T_i$ is assumed to be prefix-closed.

% $P_1\parallel P_2$ is defined to be the new process $(L_3,E_3,T_3)$, such that:
% \begin{align*}
% & L_3 = L_1 \cup L_2 \\
% & E_3 = \{ e \in E(L_3) \mid\  e \cap L_1 \in E_1 \land e \cap L_2 \in E_2 \land cond(e) \} \\
% & T_3 = \{ t\in T(E_3) \mid\ (t \project L_1) \in T_1 \land (t \project L_2) \in T_2 \}
% \end{align*}

% \red{note that i didn't put the extra conditions in the definitions of the traces $T_3$. i think they are not required, since all the events in traces of $T_3$ must be events of $E_3$, and $E_3$ imposes those conditions already.}

\ifdefined\arix
\begin{exampleregion} 
\bf{Example.} Consider a pair of processes, $P$ and $Q$, such that
$\alpha_P = \{a, b, q\}$ and $\alpha_Q = \{x, q\}$. Let $t_1\in \traces_P$
and $t_2\in \traces_Q$ be possible traces of these processes, where
$t_1=\tuple{\{a\},\{a, q\}, \{b\} }$ and $t_2=\tuple{\{x\},\{x, q\}
}$.

The following is a valid trace of the composition $P \parallel Q$:
\begin{align*}
\tuple{\{a\}, \{x\}, \{a, x, q\}, \{b\} } \in \traces_{P \parallel Q}
\end{align*}
where the event $\{a, x, q\}$ results from $P$ and $Q$
simultaneously performing $\{a,q\}$ and $\{x,q\}$, respectively.

When two events from the processes share at least one label, those
events are treated as the same kind of event, and require
simultaneous participation from both processes. Intuitively, $\{a,q\}$
may be treated like $a$, $q$, or both, 
depending on the alphabet of another process that wishes to interact with $P$. Since $Q$ is
capable of engaging in $\{x,q\}$, which itself can be treated like
$q$, $P$ and $Q$ possess an ability to influence each other through
simultaneous participation in $\{a,q\}$ and $\{x,q\}$.
In addition, when processes synchronize on a pair of events with
  distinct but overlapping sets of labels, the pair are combined into
  a new event by computing the union of the two label sets.

We also note that $cond$ is imposed on \bf{Def.\ 1} to specifically
  rule out undesirable cases in which processes interact through events in
  which they do not share any labels. For instance, the following
  trace is not a valid trace of $P \parallel Q$,
even though its projections onto
  $\alpha_P$ and $\alpha_Q$ belong into $P$ and $Q$, respectively:
\begin{align*}
\tuple{\{a, x\}, \{a, x, q\}, \{b\} } \notin \traces_{P \parallel Q}
\end{align*}
The reason the above trace is not valid is that the first event $\{a, x\}$
does not contain any common label of $P$ and $Q$.
\end{exampleregion}
% $$
% \labels_P = \{ a \in L \mid \exists \mbox{ trace } t\in P : 
% 	\exists \mbox{ event } e \in \events(t) :
% 	a \in e \}
% $$
\fi

This type of parallel composition can be seen as a generalization of the
parallel composition of CSP~\cite{hoare-csp}, from single labels to {\em sets}
of labels. That is, the CSP parallel composition is the special case of the 
composition of \bf{Def.\ 1} where
every event is a singleton (i.e., it contains exactly one label). Note
that if event $e$ contains exactly one label $a$, then $a$ must belong
to the alphabet of $P$ or that of $Q$, which means $cond(e)$ always evaluates
to true. The resulting expression in that case
\begin{align*}
\traces_{P \parallel Q} = \{ t \in T(\alpha_P \cup \alpha_Q) \mid &\ (t \project \alpha_P) \in P \land (t \project \alpha_Q) \in Q \}
\end{align*}
is equivalent to the definition of parallel composition in CSP~\cite[Sec.\ 2.3.3]{hoare-csp}.

\para{Mapping composition}

A {\em mapping $m$ over set of labels $L$} is a partial function
$m : L \fun L$. % used to introduce a relationship between a pair of distinct labels.  
Informally, $m(a) = b$ stipulates that every event
that contains $a$ as a label is to be assigned $b$ as an additional
label.
We sometimes use the notations $a \mapsto_m b$ or $(a,b)\in m$ as alternatives to $m(a) = b$.
When we write $m(a)=b$ we mean that $m(a)$ is defined and is equal to $b$.
% We also sometimes view $m$ as a relation, $m\subseteq L \times L$, and use the notation $(a,b)\in m$, with the understanding that $m$ is a function, i.e., if both $(a,b)\in m$ and $(a,c)\in m$ then $b=c$.
The {\em empty} mapping, denoted $m=\emptyset$, is the partial function
$m:L\to L$ which is undefined for all $a\in L$. 

{\em Mapping composition} allows a pair of processes to interact
with each other over distinct labels. Formally, consider two processes $P = (\alpha_P, \events_P, \traces_P)$ and $Q =
(\alpha_Q, \events_Q, \traces_Q)$, and let $L = \alpha_P\cup\alpha_Q$.  Given mapping
$m : L \fun L$, the {\em mapping composition} $P \compo Q$ is defined
as follows:
\begin{align*}
& \events_{P \compo Q} = \{ e \in E(\alpha_P \cup \alpha_Q) \mid\ e
  \cap \alpha_P \in \events_P \land e  \cap \alpha_Q \in \events_Q
  \land \\
& \qquad\qquad\qquad\qquad\qquad\qquad\quad cond'(e) \land condmap(e,m) \} \\
& \traces(P \compo  Q) = \{ t  \in (events_{P \compo Q})^* \mid (t \project \alpha_P) \in P \land (t
                     \project \alpha_Q) \in Q \}
   \tag{\bf{Def.\ 2}}
\end{align*}
where $\alpha_{P \compo Q} = \alpha_P \cup \alpha_Q$, and $cond'(e)$ and $condmap(e,m)$ are defined as:
\begin{align*}
& cond'(e) \equiv cond(e) \lor (\exists a \in e \cap \alpha_P, \exists b \in
  e \cap \alpha_Q : m(a) = b \lor m(b) = a)  \\
& condmap(e,m) \equiv (\forall a \in e, \forall b \in L :  m(a) = b \implies b \in e)
\end{align*}
The definition of $events$ in \bf{Def.\ 2} is similar to \bf{Def.\ 1}
above, but not identical:
the additional disjunct in  $cond'(e)$ allows $P$ and $Q$ to
synchronize even when they do not share any label, if at
least one pair of their labels are mapped to each other in $m$. The
predicate $condmap$ ensures that if an event $e$ contains a label $a$,
and $m$ maps $a$ to another label $b$, then $e$ also contains $b$.

Note that \bf{Def.\ 2} is different from the definition of mapping
composition in \cite{eunsuk-security}, and corrects a flaw in the
latter.  In particular, the definition in~\cite{eunsuk-security} omits
condition $cond'$, which permits the undesirable case in which events
$e_1$ and $e_2$ from $P$ and $Q$ are synchronized into union
$e=e_1\cup e_2$ even when the two events do not share any label.

Mapping composition is a generalization of parallel composition: The
latter is a special case of mapping composition where the given
mapping is empty:
\begin{lemma}
\label{lem-empty-mapping-parallel-compo}
  Given a pair of processes $P$ and $Q$, if $m = \emptyset$ then
  $P \compo Q = P \parallel Q$.
\end{lemma}
\ifdefined\arxiv
\begin{proof}
  If $m = \emptyset$, the second disjunct of $cond'(e)$ in \bf{Def.\
    2} is false, and thus $cond'(e)$ $=$ $cond(e)$. In addition,
  $condmap(e,m)$ evaluates to true, and \bf{Def.\ 2} becomes
  equivalent to \bf{Def.\ 1}.
\end{proof}
\fi

\ifdefined\arxiv
\begin{exampleregion}
  \bf{Example.} Consider two processes, $P$ and $Q$, such that
  $\alpha_P = \{a, b\}$ and $\alpha_Q = \{x, y\}$.  Let $t_1\in \traces_P$
  and $t_2\in \traces_Q$ be two traces, where $t_1=\tuple{\{b\},\{a\}}$ and
  $t_2=\tuple{\{x\},\{y\},\{x\}}$.  Suppose that we wish to compose
  $P$ and $Q$ by mapping $a$ to $x$; i.e.,\ $m(a)= x$, and $m$ is undefined
	for all other labels. This
  mapping stipulates that every event containing $a$ as a label is to
  be assigned $x$ as an additional label in the resulting composite
  process. Then, the following is {\em not} a valid trace of the
  composition (even though its projections to $\alpha_P$ and $\alpha_Q$
	give traces that are valid elements of $P$ and $Q$ respectively):
\begin{align*}
\tuple{\{b\}, \{a\}, \{x\}, \{y\}, \{x\}} \notin \traces_{P \compo Q}
\end{align*}
The reason why the above trace is not in $P \compo Q$ is that its second event, $\{a\}$, is missing $x$, even though it contains $a$ and $a$ is mapped to $x$.
The next trace is valid:
\begin{align*}
\tuple{\{b\}, \{a, x\}, \{y\}, \{x\}} \in \traces_{P \compo Q}
\end{align*}
Note that the mapping is not symmetric, which means that although it is
necessary for $a$ to synchronize with $x$, it is not necessary for
$x$ to synchronize with $a$. Thus, the fourth event in the above
trace contains $x$ without $a$. The following is also a valid trace
of the composition:
\begin{align*}
\tuple{\{b\}, \{ x\}, \{y\}, \{a, x\}} \in \traces_{P \compo Q}
\end{align*}
Here, $\{a\}$ from $t_1$ is synchronized with the
second $\{x\}$ of $t_2$, while the first $\{x\}$ of $t_2$ is allowed to
take place by itself in the composition.
\end{exampleregion}

In a mapping composition $P\compo Q$,
$P$ is typically a high-level or {\em abstract} model of an
application design, while $Q$ is a model of a low-level {\em platform}
on which $P$ is to be implemented ($P$ and $Q$ may themselves consist
of several processes). In most cases, the two processes describe
system artifacts that are built independently from each other, and
thus do not have any common labels. The mapping $m$ captures decisions
on how events from $P$ are to be realized in terms of their
counterparts in $Q$. Consequently, $P \compo Q$ describes the result of
implementing $P$ on $Q$ using these decisions.

\fi

\begin{exampleregion}
\bf{Example.} 
Let $P = \Alice \parallel \Bob \parallel \Eve$ and $Q=\Sender \parallel \Receiver$ be the abstract and public channel
communication models from Figure~\ref{fig-example}, respectively.
In this model, labels are of the form $\writeBob{m}$ or $\encWrite{m}{k}$,
where $m\in\Msg$ and $k\in\Key$ are {\em parameters}.
The domains of such parameters can be potentially infinite, resulting in 
infinite sets of labels. In this case, for the sake of simplicity,
we restrict the domains $\Msg$ and $\Key$ to be finite:
\begin{align*}
\Msg = \{ \public, \secret \}\qquad\Key = \{ \none, \keyX, \keyY \} 
\end{align*}
This, in turn, results in finite sets of labels (and thus, also a
finite number of possible mappings):
\begin{align*}
& L_P = \{  \writeBob{\public}, \writeBob{\secret}, \writeEve{\public}, \writeEve{\secret} \} \\
& L_Q = \{ \encWrite{\secret}{\none}, \encWrite{\secret}{\keyX},
                                                                                                   \encWrite{\secret}{\keyY},
  \\
& \quad\quad\quad\ \encWrite{\public}{\none}, 
\encWrite{\public}{\keyX}, \encWrite{\public}{\keyY}
\} 
\end{align*}
Suppose that we decide on a naive implementation scheme where the
abstract messages \sf{writeBob} and \sf{writeEve} are transmitted
over the public channel unencrypted. This decision can be represented
as a mapping $m_1$ where each abstract label is mapped to \sf{encWrite} with
key $k = \none$:
\begin{align*}
m_1 =
% &\writeBob{\secret}, \writeEve{\secret} \mapsto_{m_1} \encWrite{\secret}{\none} \\
% &\writeBob{\public}, \writeEve{\public} \mapsto_{m_1} \encWrite{\public}{\none}
& \{(\writeBob{\secret},\encWrite{\secret}{\none}),(\writeEve{\secret},\encWrite{\secret}{\none}), \\
&(\writeBob{\public},\encWrite{\public}{\none}), (\writeEve{\public},\encWrite{\public}{\none})\}
\end{align*}
Then, the process $\mparw{P}{Q}{m_1}$ contains traces that include the following event:
\begin{align*}
e = \{ \writeBob{\secret}, \encWrite{\secret}{\none} \}
\end{align*}
The labels of such a {\em multi-label} event correspond to its representations
at multiple abstraction layers; in this case, $e$ can be considered as both
a \sf{writeBob} and an \sf{encWrite} event. 

On the other hand, the singleton event $e' = \{ \writeBob{\secret} \}$ cannot appear in
any trace of $\mparw{P}{Q}{m_1}$, since the mapping $m_1$ requires
that each event containing $\writeBob{\secret}$ also contain
$\encWrite{\secret}{\none}$.
\end{exampleregion}

\ifdefined\arxiv
\subsection{Properties of Composition Operators}
\label{sec-properties-composition}

In this section, we state and prove various properties of the
parallel and mapping composition operators, extending prior
work~\cite{eunsuk-security} with new theoretical results.

\paragraph{Prefix closure.} Both the parallel and mapping composition operators
preserve the prefix-closure property of processes.
\begin{lemma}
\label{lem-mapping-prefix}
  Given a pair of processes $P$ and $Q$, $traces_{P \compo Q}$ is prefix-closed.
\end{lemma}
\begin{proof}
By definition, $\tuple{} \in \traces_{P \compo Q}$.

Let $t' = t \cdot \tuple{e} \in \traces_{P \compo Q}$. By the
definition of $traces$, $(t' \project \alpha_P) = (t \cdot
\tuple{e} \project \alpha_P)
\in P$. Since $P$ is prefix-closed, it must be the case that $(t
\project \alpha_P) \in
traces_P$ (similarly, $(t \project \alpha_Q) \in \traces_Q$). Thus,
it follows that $t \in \traces_{P \compo Q}$. 
\end{proof}

\begin{lemma}
\label{lem-parallel-prefix}
  Given a pair of processes $P$ and $Q$, $traces_{P \parallel Q}$ is prefix-closed.
\end{lemma}
\begin{proof}
It follows as a special case of Lemma~\ref{lem-mapping-prefix} with $m
= \emptyset$.
\end{proof}
\fi

\para{Commutativity} The proposed mapping composition operator is
commutative: i.e.,\ $\mparw{P}{Q}{m} = \mparw{Q}{P}{m}$.  This
property can be inferred from the fact that \bf{Def.\ 2} is symmetric
with respect to $P$ and $Q$. It follows that by being a special case
of mapping composition, our parallel composition operator is also
commutative.

% does not depend on which process the labels in $m$ originate from (\red{formal proof to be added}).
\ifdefined\arxiv
\para{Associativity} Our parallel composition is {\em not associative}
in general. For example, consider three processes, $P$, $Q$, and $R$,
with the following traces:
\begin{align*}
\traces_P=\{\tuple{\e{b,c}}\} \qquad \traces_Q=\{\tuple{\e{b}}\} \qquad \traces_R=\{\tuple{\e{c}}\}
\end{align*}
where $\alpha_P = \{ b, c \}$, $\alpha_Q = \{b \}$, and
$\alpha(R) = \{c\}$. If $P$ with $Q$ are first composed, and then
consequently with 
$R$, we obtain:
\begin{align*} 
\traces_{P \parallel Q} = \{\tuple{\e{b,c}} \} \quad \mbox{ and } \quad
\traces_{(P \parallel Q) \parallel R} = \{\tuple{\e{b,c}} \}
\end{align*}
However, if $Q$ and $R$ are composed first, then
$\e{b}$ and $\e{c}$ are interleaved as two distinct events.
Consequently, when $(Q \parallel R)$ is further composed with $P$, $\e{b,c}$
cannot take place in the final composition:
\begin{align*}
\traces_{Q\parallel R} = \{\tuple{\e{b}, \e{c}}, \tuple{\e{c}, \e{b}}\}
\quad \mbox{ and } \quad
\traces_{P \parallel (Q \parallel R)} = \emptyset
\end{align*}
\else
\fi

\ifdefined\arxiv
The mapping composition $\compo$ is also not associative in general,
since $\parallel$ itself is not associative; in fact, $\compo$ is
not associative in general even when the mapping $m$ is not empty.
For example, consider three processes, $P$, $Q$, and $R$, with the
traces:
\begin{align*}
\traces_P = \{ \tuple{\{a\}} \} \qquad \traces_Q = \{ \tuple{\{b\}}\} \qquad \traces(R) = \{\tuple{\{c\}}\} 
\end{align*}
where $\alpha_P = \{a\}$, $\alpha_Q = \{b\}$, and $\alpha(R) =
\{c\}$.
% Suppose that we wish to compose these processes using two distinct mappings, $m_1$, and $m_2$, with the intention that every $a$-event be assigned both $b$ and $c$ as additional labels in the final process:
Let 
% \begin{align*}
$m_1 = \{(a,b)\}$ and $m_2 = \{(a, c)\}$.
% \end{align*}
Then, 
$
\mparw{(\mparw{P}{Q}{m_1})}{R}{m_2} \neq \mparw{P}{(\mparw{Q}{R}{m_2})}{m_1}.
$
On the right hand side, when $Q$ and $R$ are composed under $m_2$,
$\{b\}$ and $\{c\}$ are interleaved as separate events in the
resulting process, since neither $b$ or $c$ appears in the domain of
$m_2$:
\begin{align*}
\traces(\mparw{Q}{R}{m_2}) = \{\tuple{\e{b}, \e{c}}, \tuple{\e{c}, \e{b}}\}
\end{align*}
When we further compose this process with $P$ under $m_1$,
the resulting process will contain traces where $a$ is merged
with $b$; e.g.,
\begin{align*}
\traces(\mparw{P}{(\mparw{Q}{R}{m_2})}{m_1}) = \{ \tuple{\{a,b\}, \{c\}}, \tuple{\{c\}, \{a,b\}} \}   
\end{align*}
On the other hand, if $P$ is first composed with $Q$ under $m_1$, and then
subsequently with $R$ under $m_2$, the resulting process has a trace where
an event contains all of $a$, $b$, and $c$:
\begin{align*}
\traces(\mparw{P}{Q}{m_1}) = \{ \tuple{a, b} \} \qquad \traces(\mparw{(\mparw{P}{Q}{m_1})}{R}{m_2}) = \{\tuple{\{a,b, c\} } \}
\end{align*}

Despite the above negative general results, we can show that our
composition operators are associative under certain conditions on the
alphabets of involved processes and mappings.  \else
\para{Associativity} Our mapping composition is associative under
the following conditions on the alphabets of involved processes and
mappings.
\fi
\begin{theorem}
\label{theorem-mapping-associativity}
Given processes $P$, $Q$, and $R$, let $X = \mparw{(\mparw{P}{Q}{m_1})}{R}{m_2}$ and $Y = \mparw{P}{(\mparw{Q}{R}{m_4})}{m_3}$. If $\events(X) = \events(Y)$,
then $X=Y$.
\end{theorem}
\ifdefined\arxiv
\begin{proof}
Found in Appendix~\ref{appen-thm-associativity}.
\end{proof}

\red{Stavros: what about the public channel example? does it satisfy these conditions? If not, then what does $\Alice \parallel \Bob \parallel \Eve$ mean?
If yes, is it because they only have single labels? We need then to state
an extra theorem which states associativity for processes with single label
events.}
\else
\begin{proof}
Available in the extended version of this paper~\cite{mapping-arxiv}.
%\red{upload new arxiv version}
\end{proof}
\fi
\ifdefined\techreport
\input{operational-semantics}
\fi

%% file: mapping-verification.tex
\section{Verification Problems}
\label{sec-mapping-verification}

\ifdefined\arxiv
Our framework allows to formulate several interesting problems related
to mappings. One such problem is {\em mapping verification}: given
$P,Q,m$, and a property $\phi$, check whether $P\compo Q$ satisfies
$\phi$.  This verification problem has already been studied
in~\cite{eunsuk-security}, although not formulated explicitly as such,
and not formalized completely.  Here, we propose a complete
formalization of the mapping verification problem, and in particular
clarify the relation between the specification of properties at the
high-level design and at the low-level implementation.  At the same
time, we prepare the ground for the formulation of the synthesis
problems that follow in Section~\ref{sec-synthesis-problems}.

\subsection{Specification}

Before defining the mapping verification problem, let us explain how
our framework allows the specification of properties at different
levels of abstraction seamlessly.  Recall that the goal of our
methodology is to allow the system designer to express a desired
property on a high-level design $P$ and preserve the same property as
the design is implemented on a low-level platform ($P \compo Q$).
Therefore, ideally, we would like to be able to do the following: (1)
Model the high-level design as process $P$.  (2) Express the desired
properties as a specification $\phi$.  (3) Check that the design
satisfies its specification, i.e., check that $P\models \phi$.  (4) Model
the execution platform as process $Q$.  (5) Express the implementation
decisions as a mapping $m$.  (6) Check that the original
implementation is not violated in the implemented system, i.e., check
that $(P \compo Q)\models \phi$.

Unfortunately, we cannot immediately follow the steps above and
``reuse'' the original specification $\phi$ in the mapping verification
problem. The reason is that the high-level process $P$ contains traces
which are of a different type from the traces of the implementation
model $P \compo Q$. Indeed, $P$ is a subset of $T(2^{\alpha(P)})$,
whereas $P \compo Q$ is a subset of
$T(2^{\alpha(P)\cup\alpha(Q)})$. Also, the specification $\phi$ should
ideally be {\em implementation independent}, which means that $\phi$
should be of the same type as $P$, i.e., a set of traces in
$T(2^{\alpha(P)})$. Therefore, $\phi$ is a-priori incomparable with
$P \compo Q$, and the problem of checking $(P \compo Q)\models \phi$
becomes ill-defined.  \red{one may ask here: why not simply project
  $P \compo Q$ to $\alpha(P)$, i.e., instead of checking
  $(P \compo Q)\models \phi$, check that
  $((P \compo Q) \project \alpha(P))\models \phi$. then perhaps we don't
  need the discussion below? (Eunsuk: In many cases, yes, this is what
  we want, but not in general. For example, we might have a spec that
  talks about the cardinality of events -- e.g., ``event with label x
  cannot contain another label''. We want to be able to model
  systems where this type of spec is violated by a mapping; and
  the projection solution doesn't achieve this.)}

We take the following approach to resolve this incompatibility. In
this paper, we are interested in specifications that are expressible
as a \textit{trace property}~\cite{schneider-liveness} (i.e., a set of
traces describing desirable system behaviors). An initial property
specified over $P$ would contain traces that are expressed only in
terms of the alphabet of $P$; let us label this property $X$. Once $P$
is realized as part of $Q$ through mapping $m$, some subset of the
traces of the resulting process ($P \compo Q$) would capture behaviors
that are consistent with $X$ and expressed over the alphabets of both
$P$ and $Q$; let us refer to the largest such subset as $X'$. In other
words, $X'$ is the property that captures the designer's intent in the
platform implementation $P \compo Q$.

Our goal is to allow the designer to specify only an initial
specification that yields $X$, and have $X'$ derived \textit{automatically}
from the same specification when $P$ is mapped to $Q$. To achieve
this, we treat \textit{specification} $\phi$ as a syntactic object that
is constructed in terms of some given set of labels $L_{\phi}$. To obtain
its meaning as a property, $\phi$ is evaluated with respect to some set
of labels $L$. The evaluation function, $sem$ is defined as follows:
\begin{align*}
sem(\phi, L) = 
\begin{cases}
\{ t \in T(L) | sat(t, \phi) \} & \text {if } L_{\phi} \subseteq L \\
\text{undefined} &   \text{ otherwise}
\end{cases}
\end{align*}
where $sat$ returns true if and only if trace $t$ satisfies the
specification $\phi$. Note that $sem$ requires only that the input set
$L$ contains at least those labels over which $\phi$ is expressed; in
other word, $L$ may contain labels that are not mentioned by $\phi$. This
allows us to construct properties over multiple levels of system
abstraction from a single specification $\phi$:
\begin{align*}
X = sem(\phi, \alpha(P)) \qquad X' = sem(\phi, \alpha(P \compo Q))
\end{align*}
The meaning of $sat$ itself depends on the type of specification
$\phi$. For instance, if $\phi$ is a first-order logic predicate with free
variable $t$ that corresponds to trace $t$, then we can define the
meaning of $sat$ simply as:
\begin{align*}
sat(t, \phi) \iff \phi(t)
\end{align*}

% For instance, suppose that $\phi$ is an LTL formula defined over
% some set of proposition variables $AP$, and
% $f : T \rightarrow (2^{AP})^*$ is a function that maps each trace to a
% sequence of sets of propositions, where each element in $f(t)$
% captures the propositions that are true at a particular point in the
% execution trace $t$. Then $sat(t,S)$ can be defined as evaluating to
% true if and only if $f(t)$ satisfies $\phi$ according to the semantics of
% LTL~\cite{pnueli-ltl} (\todo{We are cheating a little bit
%   here, since LTL is defined over infinite sequences}).

Given the above notion of specification, we can now define what it
means for a system, modeled as a process, to satisfy its
specification:
\begin{align*}
P \models \phi \iff P \subseteq sem(\phi,\alpha(P))
\end{align*}
In other words, the traces allowed by $P$ must be a subset of the
property that is constructed from $\phi$ and the alphabet of $P$.

\begin{exampleregion}
  \bf{Example.} Recall the example from Figure~\ref{fig-example},
  where the desired specification that \sf{Eve} never learns the
  secret of \sf{Alice} may be stated as the first-order logic formula
\begin{align*}
\neg \exists e \in events(t) : \sf{writeEve}(\secret)  \in e
\end{align*}
where $t$ represents a trace over the set of all labels $L$. Let $\phi$
refer to this formula, $P$ the abstract model in
Figure~\ref{fig-example}(a), and $Q$ the concrete model in
Figure~\ref{fig-example}(b). Then, the following trace (of length 2),
which describes \Alice sending a secret to \Bob, and then a public
message to \Eve, is a valid behavior according to the specification,
since it does not result in \Eve learning the secret:
\begin{align*}
\tuple{\e{\sf{writeBob(secret)}}, \e{\sf{writeEve(public)}}}\in sem(\phi, \alpha(P))
\end{align*}
Similarly, the following trace, expressed over the labels that appear
in $P$ and $Q$, is also a valid behavior as defined by $\phi$:
\begin{align*}
\tuple{\e{\sf{writeBob(secret)}, \encWrite{\secret}{\keyX}}, 
  \e{\sf{writeEve(public)}, \encWrite{\public}{\none}}}\in sem(\phi, \alpha(P \compo Q))
\end{align*}
because executing this sequence of events does not lead to state \sf{SecretLearned}.
\end{exampleregion}

\subsection{Verification} 
\fi

In formal methods, a typical verification problem involves, given a
process $P$ and a property $\phi$, checking whether $P$ satisfies $\phi$
(i.e., $P\models \phi$).  In our framework, an additional verification
problem is defined as follows:
\begin{problem}[Mapping Verification]
Given processes $P$ and $Q$, mapping $m$, and property $\phi$, check whether
$(P \compo Q) \models \phi$.
\end{problem}
We say a mapping $m$ is {\em valid} (w.r.t. $P$, $Q$ and $\phi$) if and
only if $(P \compo Q)\models \phi$. When solving the mapping verification
problem, we may rely on the assumption that $P$ alone satisfies the
specification, i.e., that $P\models \phi$ holds (which could be checked
beforehand). Note that this problem is conceptually an
$\exists$-problem, as it can be stated as finding a witness trace $t$
to formula $\exists t : t \in (P \compo Q) \land t\not\in \phi$.

\ifdefined\arxiv
Such problems can be answered in a non-exhaustive manner using
satisfiability solvers. By ``non-exhaustive'' we mean that when the
state-space is infinite, the solver typically gives an incomplete
answer, by exploring only a finite subset of the state-space. 

Another type of analysis problem, also introduced
in~\cite{eunsuk-security}, can be stated on mappings as follows:
\begin{problem}[Insecure Mapping Search]
Given processes $P$ and $Q$, and property $\phi$, find a mapping $m$ such that
($P \compo Q) \not\models \phi$.
\end{problem}
Such a mapping $m$, if it exists, describes decisions that may introduce
undesirable behavior into the resulting implementation, undermining a
property $\phi$ that has been previously established by $P$. Note that
this problem is a {\em search} rather than a {\em synthesis} problem,
since it can be formulated as an $\exists$-problem:
$\exists m, \exists t : t \in (P \compo Q) \land t\not\in \phi$.
 \fi

\begin{exampleregion}
\bf{Example.} Recall the mapping $m_1$ from
Section~\ref{sec-modeling-framework}, where abstract messages are
transmitted as unencrypted public messages. Let $\Alice'$ and $\Eve'$
be processes that represent the implementation of $\Alice$ as
$\Sender$ using $m_1$ (i.e.,\ $\Alice \|_{m_1}\Sender$), and $\Eve$ as
$\Receiver$, respectively. The following $t_1$ and $t_2$ are valid traces of these two
processes:
\begin{align*}
t_1 = \tuple{\{\writeBob{\secret}, \encWrite{\secret}{\none}\} } \in \Alice' \\
t_2 = \tuple{\{\writeEve{\secret}, \encWrite{\secret}{\none}\} } \in \Eve'
\end{align*}
When composed in parallel, $\Alice'$ and $\Eve'$ may interact through
the above two events, since they share the common label
$\sf{encWrite}$ $(\secret,\none)$. According to the definition (\bf{1})
from Section~\ref{sec-modeling-framework}, the following is a valid
trace of their parallel composition:
\begin{align*}
&\tuple{
\{\writeBob{\secret}, \writeEve{\secret}, \encWrite{\secret}{\none}\}} \in (\Alice' \| \Eve')
\end{align*}
Note that in Figure~\ref{fig-example}(a), this event results in the
transition of \Eve into state \sf{SecretLearned}, leading to a
violation of the desired property.  This
can be seen as an example of abstraction violation: As a result of
design decisions in $m_1$, $\sf{writeBob}$ and $\sf{writeEve}$ now share
the same underlying representation ($\sf{encWrite}$), and $\Eve'$ is now
able to engage in an event that was not previously available to it in
the original abstract model.
\end{exampleregion}

%% file: mapping-synthesis.tex
\section{Synthesis Problems}
\label{sec-synthesis-problems}

\input{synthesis-problems}

%% file: synthesis-problems.tex
%\todo{Eunsuk:Say it's the first definition of the synthesis problem for mapping}

% \section{Synthesis Problems}

In this paper, we propose the novel problem of \emph{synthesizing} a mapping
that preserves a desired property $\phi$ between $P$ and $Q$.
\begin{problem}[\bf{Mapping Synthesis}]
Given processes $P$ and $Q$, and property $\phi$, find, if it exists,
a mapping $m$ such that $(P \compo Q) \models \phi$.
\label{problem-mapping-synthesis}
\end{problem}
Note that this problem can be stated as a $\exists\forall$ problem;
that is, finding a witness $m$ to the formula
$\exists m : \forall t : t \in (P \compo Q) \implies t \in \phi$.

% A \textit{partial
%   mapping} is a pair $(m_p,m_n)$, where $m_p$ is a partial function
% that describes a minimum required set of tuples, and $m_n$ is a set of
% negative tuples that should never appear in any $m$.  Then, a
% completion $m$ is a mapping such that $m_p \subseteq m$ and
% $m \cap m_n=\emptyset$.
% \begin{problem}[Mapping Completion]
%   Given processes $P$ and $Q$, property $\phi$, and partial mapping
%   $(m_p,m_n)$, find, if it exists, a mapping $m$ such that
%   $m_p \subseteq m$, $m \cap m_n=\emptyset$, and
%   $(P \compo Q)\models \phi$.
% \label{problem-mapping-completion}
% \end{problem}
% Informally, $m_p$ can be regarded as given design decisions that must
% be included as part of the final implementation. On the other hand,
% $m_n$ describes restrictions that are intended to exclude certain
% design decisions from being considered as part of a valid
% implementation. Note that the mapping synthesis problem is a special
% case of mapping completion where both $m_p$ and $m_n$ are empty.

% In practice, however, the developer may wish to specify a more general
% form of partial mappings that express uncertainty or choice among
% implementation decisions. 

Instead of synthesizing $m$ from scratch, the developer may often wish
to express her partial system knowledge as a given
\textit{constraint}, and ask the synthesis tool to generate a mapping
that adheres to this constraint. For instance, given labels
$a, b, c \in L$, one may express a constraint that $a$ must be mapped
to either $b$ or $c$ as part of every valid mapping; this gives rise
to two possible candidate mappings, $m_{1}$ and $m_{2}$, where
$m_{1}(a) = b$ and $m_{2}(a) = c$. Formally, let $M$ be the set of all
possible mappings between labels $L$. A \textit{mapping constraint}
$C \subseteq M$ is a set of mappings that are considered legal
candidates for a final, synthesized valid mapping.  Then, the problem
of synthesizing a mapping given a constraint can be formulated as
follows:
\begin{problem}[\bf{Generalized Mapping Synthesis}]
  Given $P$ and $Q$, property $\phi$, and mapping constraint
  $C$, find, if it exists, a mapping $m$ such that $m \in C$ and
  $(P \compo Q) \models \phi$.\label{problem-generalized-mapping-synthesis}
\end{problem}
Note that Problem~\ref{problem-mapping-synthesis} is a special case of
Problem~\ref{problem-generalized-mapping-synthesis} where $C = M$.
\ifdefined\arxiv
\begin{exampleregion}
\bf{Example.} Consider a pair of processes, $P$ and $Q$, such that
$\alpha(P) = \{a, b \}$ and $\alpha(Q) = \{x, y, z\}$. Suppose that the
developer wishes to map $a$ to either $x$ or $y$, and $b$ to $z$; this
can be specified as the following constraint:
\begin{align*}
C = \{m_1, m_2\} \quad m_1 = \{ (a, x), (b, z) \} \quad m_2 = \{ (a, y),
  (b, z)\}
\end{align*}
Suppose instead that the developer's intent is to map $a$ to any one
of the labels in $Q$, but not map $b$ to anything; this can be
specified as
\begin{align*}
C = \{m_1, m_2, m_3\} \quad m_1 = \{ (a, x) \} \quad m_2 = \{ (a, y) \}
  \quad m_3 = \{ (a, z) \}
\end{align*}
\end{exampleregion}
\fi
The synthesis problem can be further
generalized to one that involves synthesizing a constraint that contains
a \textit{set} of valid mappings:
\begin{problem}[\bf{Mapping Constraint Synthesis}]
  Given $P$ and $Q$, property $\phi$, and mapping
  constraint $C$, generate a non-empty constraint $C'$ such that
  $C' \subseteq C$, and for every $m \in C'$,
  $(P \compo Q) \models \phi$.
\label{problem-constraint-synthesis}
\end{problem}
If it exists, $C'$ yields a set of valid mappings that adhere to the
given constraint $C$ (i.e., $C'$ is a stronger constraint than
$C$). We call such constraint $C'$ \textit{valid} with respect to $P$,
$Q$, $\phi$ and $C$. A procedure for solving
Problem~\ref{problem-constraint-synthesis} can in turn be used to
solve Problem~\ref{problem-generalized-mapping-synthesis}: Having
generated constraint $C'$, we pick some mapping $m$ such that
$m \in C'$.

% Let us state the following lemma to define what it means for
% constraint $C'$ to
% be stronger than $C$.
% \begin{lemma}
% For any pair of constraints, $C$ and $C'$, the following holds:
% \begin{align*}
% (M_{C'} \subseteq M_C) \iff \forall a \in L : C'(a)
%   \subseteq C(a) \land (C(a) \neq \emptyset \implies C'(a) \neq \emptyset)
% \end{align*}
% \end{lemma}
% \begin{proof}
% ($\implies$) direction: Proof by contradiction. Suppose that there
% exists some $a$ such that $C'(a) \nsubseteq C(a)$ (Case 1) or $C(a) \neq
% \emptyset \land C'(a) = \emptyset$ (Case 2).

% In Case 1, let $b \in
% C'(a) - C(a)$. By definition \bf{(3)}, there exists at least one
% mapping $m \in M_{C'}$ such that $m(a) = b$. By the same definition,
% $m$ cannot be a mapping in $M_C$, leading to a contradiction. 

% In Case 2, according to definition \bf{(3)}, there exists at least one
% mapping $m \in M_{C'}$ such that $a \notin dom(m)$. Since $C(a)$ is
% non-empty, by the same definition, there exists no mapping
% $m' \in M_{C}$ such that $a \notin dom(m')$, thus resulting in a
% contradiction.

% ($\Leftarrow$) direction: Consider $m \in M_{C'}$. For each $a \in L$,
% if $C(a)$ is empty, then it follows that $C'(a)$ is also empty, and
% $a \notin dom(m)$.  If, on the other hand, $C(a)$ is non-empty, then
% $C'(a)$ is also non-empty, and $m$ maps $a$ to some label $b$ such
% that $b \in C(a)$. By definition \bf{(3)}, it follows that $m \in
% M_C$. 
% \end{proof}

In practice, it is desirable for a synthesized constraint to be as
large as possible while still being valid, as it provides more
implementation choices (i.e., possible mappings). The
problem of synthesizing a \emph{maximal} mapping constraint is defined as:
\begin{problem}[\bf{Maximal Constraint Synthesis}]
  % Given a pair of processes $P$ and $Q$, a property $\phi$, and a mapping
  % constraint $C$, generate a constraint $C'$ such that for every $a
  % \in L$, $C'(a) \subseteq C(a)$, and for every $m \in M_{C'}$,
  % $\mparw{P}{Q}{m} \models \phi$.
  Given processes $P$ and $Q$, property $\phi$, and given
  constraint $C$, generate a constraint $C'$ such (1) $C'$ is valid
  with respect to $P$, $Q$, $\phi$, and $C$,
  and (2) there exists no other constraint $C''$ such that $C''$ is valid and
  $C' \subseteq C''$. 
\label{problem-optimal-constraint-synthesis}
\end{problem}
Note that $C'$, if found, is a \textit{local} optimum. In general, there may be
multiple local optimal constraints for given $P$, $Q$, $\phi$, and $C$.

\begin{exampleregion} \textbf{Example.} In our running example, suppose that there are two \Receiver
  processes with the identical behavior ($\ReceiverX$ and
  $\ReceiverY$), except that they are assigned unique decryption keys:
  $\knows(\ReceiverX)$ $=$ $\{\keyX\}$ and $\knows(\ReceiverY)$ $=$ $\{\keyY\}$.
Then, the following $m_2$ is a valid mapping
that preserves the desired specification that $\Eve$ never learns the secret:
\begin{align*}
&\writeBob{\secret} \mapsto_{m_2} \encWrite{\secret}{\keyX}, \ \writeBob{\public} \mapsto_{m_2} \encWrite{\public}{\keyX} \\
&\writeEve{\secret} \mapsto_{m_2} \encWrite{\secret}{\keyX}, \ \writeEve{\public} \mapsto_{m_2} \encWrite{\public}{\keyY} 
\end{align*}
That is, if the secret message are always encrypted with the key assigned
to $\Bob$, $\Eve$ will never be able to read it. The following $m_3$
is also a a valid mapping:
\begin{align*}
&\writeBob{\secret} \mapsto_{m_3} \encWrite{\secret}{\keyX}, \ \writeBob{\public} \mapsto_{m_3} \encWrite{\public}{\none} \\
&\writeEve{\secret} \mapsto_{m_3} \encWrite{\secret}{\keyX}, \ \writeEve{\public} \mapsto_{m_3} \encWrite{\public}{\none} 
\end{align*}
since $\Eve$ being able to read public messages does not violate the
property. Thus, the developer may choose either $m_2$ or $m_3$ to
implement the abstract channel and ensure that $\secret$ remains
protected from $\Eve$. In other words, $C_1 = \{ m_2, m_3 \}$ is a
valid (but not necessarily maximal) mapping constraint with respect to
the desired property. Furthermore, $C_1$ is arguably more desirable
than another constraint $C_2 = \{m_2\}$, since the former gives the
developer more implementation choices than the latter does.
\end{exampleregion}

%% file: synthesis-technique.tex
\section{Synthesis Technique}
\label{sec-synthesis-technique}

% This section presents a method for representing mappings
% \it{symbolically}, and a technique for synthesizing optimal mapping
% constraints. 

\ifdefined\techreport
\subsection{Symbolic Mapping Representation}
\label{sec-mapping-representation}
\else \para{Mapping representation}
\fi
Assuming that the set of event labels $L$ is finite, one possible way
to represent a mapping is by \textit{explicitly} listing all of the
entries in the function. An alternative representation is one where
mappings are represented \textit{symbolically} as logical expressions
over variables that correspond to labels being mapped. The
symbolic representation has the following advantages: (1) it provides
a succinct representation of implementation decisions to the developer
(which is especially important as the size of the mapping grows large)
and (2) it allows the user to specify partial decisions
(i.e.,\ given constraint $C$) in a \textit{declarative} manner.

In particular, we use a \textit{grammar-based} approach, where the
space of candidate mapping constraints is restricted to expressions
that can be constructed using a syntactic grammar~\cite{sygus}. Our
grammar is described as follows\ifdefined\techreport\footnote{However, our synthesis algorithm is not tied to a
  particular grammar, and developers may construct their own grammar based on
  the insights from the problem domain.}\fi:
\begin{align*}
& Term := Var | Const \qquad Assign := (Term = Term)  \\
    & Expr := Assign | \neg Assign | Assign \implies Assign | Expr \land Expr
  \end{align*}
  where $Var$ is a set of variables that represent parameters inside a
  label, and $Const$ is the set of constant values.  Intuitively, this
  grammar captures implementation decisions that involve assignments
  of parameters in an abstract label to their counterparts in a
  concrete label (represented by the equality operator ``$=$''). A
  logical implication is used to construct a conditional assignment of
  a parameter. For example, expression $\mathcal{X}_{\textsf{writeBob}}(a, b)$ $\equiv$
  $a.m = b.m$ $\land$ $(b.m = \secret \implies b.k = \keyX)$ states
  that if the message being transmitted is a secret, key $\keyX$ must
  be used ($a$ and $b$ are variables representing 
  \sf{writeBob} and \sf{encWrite} labels, respectively).
% Suppose instead that $a$ is to be mapped to one of $x$, $y$, and $z$,
% and $b$ to none. One symbolic expression that may be used to represent
% this constraint is:
% \begin{align*}
% X(v, w) \equiv ((v = a \land w = x) \lor (v = a \land w = y) \lor (v =
%   a \land w = z)) \land \neg (v = b)
% \end{align*}

In general, given an expression, the constraint that it represents is
computed as:
\begin{align*}
C = \{ m : L \fun L | \forall a \in L : & (\exists b \in L : \mathcal{X}(a, b))
  \implies \mathcal{X}(a, m(a)) \land \\ 
& \neg (\exists b \in L : \mathcal{X}(a, b)) \implies a \notin dom(m) \}
\end{align*}
That is, a mapping $m$ is allowed by $C$ if and only if for each
label $a$, (1) if there is at least another label $b$ for which
$\mathcal{X}(a,b)$ evaluates to true, then $m$ maps $a$ to one of such labels,
and (2) otherwise, $m$ is not defined over $a$.

\ifdefined\techreport
In particular, we use a \textit{grammar-based} approach, where the
space of candidate mapping constraints is restricted to expressions
that can be constructed using a syntactic grammar~\cite{sygus}. For
the systems that we have studied, we use the following
grammar\footnote{However, our synthesis algorithm is not tied to a
  particular grammar, and developers may construct their own grammar based on
  the insights from the problem domain.}:
\begin{align*}
    & Term := Var | Const \qquad Assign := (Term = Term) \\
    & Expr := Assign | Assign \implies Assign | Expr \land Expr
  \end{align*}
  where $Var$ is a set of variables that represent parameters inside a
  label, and $Const$ is the set of constant values.  Intuitively,
  this grammar captures implementation decisions that involve
  assignments of parameters in an abstract label to their counterparts
  in a concrete label (represented by the equality operator
  ``$=$''). A logical implication is used to construct a conditional
  assignment of a parameter.

Each expression $\mathcal{X} \in Expr$ is a formula that describes a
relationship between a pair of labels (e.g.,\ $(a,b)$) that may appear in a
mapping (i.e.,\ $m(a) = b$). To
define what it means for a label pair to satisfy an expression, let us
introduce the notion of \textit{signatures}, each of which describes
the structure of a particular group of labels. Formally, a signature
$\sigma$ is a tuple $(params, typ)$ where $params \in \power(Var)$ are
variables representing label parameters, and $typ : Var \fun Typ$ is a
function that maps each variable to its type.

Each label $a \in L$ is associated with exactly one signature (denoted
$sig(a)$). For instance, given label $a = \encWrite{\secret}{\keyX}$,
$sig(a) = \sigma_{\sf{encWrite}}$ is the signature that describes the
structure of all \sf{encWrite} labels. In particular, here
$params= \{m, k\}$ and $typ = \{ (m, \sf{Msg}), (k, \sf{Key})\}$,
where $m$ and $k$ are variables that represent the message and key
associated with each \sf{encWrite} label. Furthermore, every label
$a \in L$ assigns a concrete value to each parameter that appears in
its signature. For example, given label
$a = \encWrite{\secret}{\keyX}$, $a.m = \secret$, and $a.k = \keyX$,
where \sf{secret} and \sf{keyX} are constant values that represent a
secret message and a particular key\footnote{For shorthand, we use the
  notation $a.v$ to refer to the variable $v$ of $sig(a)$.}.

Separate expressions may be used to describe relationships between
labels of different signatures. Given a pair of labels $a$ and $b$,
let $\sigma_a$ and $\sigma_b$ be their respective signatures, and let
$\mathcal{X}_{(\sigma_a,\sigma_b)}$ be an expression that describes a
relationship between parameters of these signatures. Then, a mapping
constraint $C$ can be defined in terms of expressions as follows:
\begin{align*}
% C = \{ m : L \fun L | \forall a \in L : & (\exists b \in L : \mathcal{X}(a, b))                                        
%   \implies \mathcal{X}(a, m(a)) \land \\ 
% & \neg (\exists b \in L : \mathcal{X}(a, b)) \implies a \notin dom(m) \} 
C = \{ m : L \fun L | \forall a \in L :  &\  (\exists b \in L : E_{(\sigma_a,\sigma_b)}(a, b))                                         
  \implies \mathcal{X}_{(\sigma_a,\sigma_b)}(a, m(a)) \land \\ 
& \neg (\exists b \in L : \mathcal{X}_{(\sigma_a,\sigma_b)}(a, b)) \implies a \notin dom(m) \} 
\end{align*}
That is, if expression $\mathcal{X}_{(\sigma_a,\sigma_b)}$ evaluates to true
over some pair $(a, b)$, then each $m \in C$ must be defined on $a$, and
the pair $(a, m(a))$ must also satisfy $\mathcal{X}$; otherwise, $m$ cannot be
defined over $a$ (i.e., the expression does not allow $a$ to be mapped
to any other label).
\fi

%  to check whether pair $(a, b)$ appearing in a mapping
% satisfies $C$, we look up the corresponding expression that is
% expressed over the parameters of $\sigma_a$ and $\sigma_b$. 

% Let us now define what it means for a pair of labels $(a, b)$ in
% mapping $m$ to satisfy constraint $C$. Let
% $params : L \fun \power(Var)$ be a function that maps each label to a
% set of variables that correspond to the parameters associated that
% label. For example, given an \sf{encWrite} label $l$,
% $params(l) = \{m, k\}$ where $m$ and $k$ are variables that correspond
% to the message and key; for shorthand, we use the notation $l.m$ to
% refer to the variable $m$ of label $l$. Each label $l$ is also
% associated with an \textit{assignment function}
% $I_l : Var \rightarrow Const$, which maps every variable associated
% with $l$ to some constant (where $dom(I_l) = params(l)$). For example,
% given label $l = \encWrite{\secret}{\keyX}$, $I_l$ is defined as
% $\{ (l.m, \secret), (l.k, \keyX) \}$.

\begin{exampleregion}
\bf{Example.} In our running example, the signatures associated with
event labels are as follows:
\begin{align*}
\sf{writeBob}(m: \Msg)\ \sf{writeEve}(m: \Msg)\ \sf{encWrite}(m:
  \Msg, k: \Key)
\end{align*}
Separate constraint expressions are used to capture the mapping
between \sf{writeBob} and \sf{encWrite}, and \sf{writeEve} and
\sf{encWrite}; we will call these expressions $\mathcal{X}_{\sf{writeBob}}$ and
$\mathcal{X}_{\sf{writeEve}}$, respectively.  One possible mapping constraint
that preserves the secrecy property may be expressed as follows:
\begin{align*}
E_{\sf{writeBob}} (a, b) \equiv\ & a.m = b.m \land (b.m = \secret \implies b.k = \keyX) \\
E_{\sf{writeEve}} (a, b) \equiv\ & a.m = b.m \land (b.m = \secret
                                   \implies b.k = \keyX)
\end{align*}
Informally, these expressions stipulate that (1) the messages
associated with the high-level and low-level events are identical, and
(2) if the message being transmitted is a secret, the key associated
with $\Bob'$ must be used. Note that this constraint does not specify
which keys should used to encode public messages, since this decision
is not relevant to the secrecy property. This under-specification is
desirable, since it gives more implementation freedom to the developer
while preserving the desired property at the same time.
\end{exampleregion}

\ifdefined\techreport
\subsection{Algorithm}
\label{sec-algorithm}

In this section, we describe an algorithm that produces a maximal valid constraint as a solution to
Problem~\ref{problem-optimal-constraint-synthesis} in
Section~\ref{sec-synthesis-problems}.
\else
\para{Algorithmic considerations}
\fi
% To ensure that the algorithm terminates, we assume that the set of
% expressions that may be constructed using a given grammar is
% restricted to a finite set. This restriction can be enforced either by
% bounding the size of expressions, or finitizing the domains of data
% types (\sf{Msg} and \sf{Key} in our example). The specifics of the
% restriction are provided by the user as an input.

% \begin{figure}[!t]
% \centering
% \includegraphics[width=0.49\textwidth]{diagrams/algorithm}
% \caption{An outline of the synthesis algorithm.}
% \label{fig-algorithm}
% \end{figure}

% We use a simple enumerative approach to synthesis, where we
% systematically enumerate candidate mapping constraints using a given
% grammar, and check each candidate for validity using a
% verifier. Our algorithm assumes the existence of a tool with an input
% language that is capable of expressing the mapping composition
% introduced in Section~\ref{sec-modeling-framework}, and the mapping
% verification problem (Problem 4) in
% Section~\ref{sec-mapping-verification}.
% The outline of our synthesis algorithm is shown in
% Figure~\ref{fig-algorithm}.

To ensure that the algorithm terminates, the set of expressions that
may be constructed using the given grammar is restricted to a finite
set, by bounding the domains of data types (e.g.,\ \sf{Msg} and
\sf{Key} in our example) and the size of expressions.  We
also assume the existence of a verifier that is capable of checking
whether a candidate mapping satisfies a given property $\phi$. The
verifier is used to implement function $verify(C,P,Q,\phi)$, which
returns $OK$ if and only if every mapping allowed by constraint $C$
satisfies $\phi$.
% In
% Section~\ref{sec-implementation}, we describe one instantiation of
% this synthesis algorithm using the Alloy Analyzer, a formal modeling
% and analysis tool~\cite{alloy}.

\para{Naive algorithm} Once we limit the number of candidate
expressions to be finite, we can use a brute-force algorithm to
enumerate and check those candidates one by one. Unfortunately, this
naive algorithm is likely to suffer from scalability issues, as our
experimental results indeed indicate (see
Section~\ref{sec-case-study}, Figure~\ref{fig-results}).

% Given that the number of candidate expressions
% is finite, one simple algorithm involves simply enumerating and
% verifying the candidates one-by-one. This algorithm, although simple
% to implement, is likely to suffer from scalability issues for problems
% with a large number of candidates. Indeed, we attempted to apply this
% algorithm to our OAuth case study (described later in
% Section~\ref{sec-case-study}), and discovered that enumerating the
% entire search space under a reasonable amount of time is not
% feasible.

\para{Generalization algorithm} We present an improved algorithm that
takes a generalization-based approach to dynamically identify and
prune undesirable parts of the search space. A
key insight is that only a few implementation decisions---captured by
some \textit{minimal subset} of the entries in a mapping---may be
sufficient to imply that the resulting implementation will be
invalid. Thus, given some invalid mapping, the algorithm
attempts to identify this minimal subset and construct a larger
constraint $C_{bad}$ that is guaranteed to contain only invalid
mappings.

\ifdefined\techreport
\begin{figure}[t]
\centering
\includegraphics[width=0.30\textwidth]{diagrams/maximization.pdf}
\caption{\small{An illustration of generalization of $C$ to
    $C_{bad}$. The red circles represent regions containing only invalid
    mappings.}}
\label{fig-generalization}
\end{figure}
\fi

% \begin{wrapfigure}{r}{0.3\textwidth}
% \begin{center}
% \includegraphics[width=0.30\textwidth]{diagrams/maximization.pdf}
% \end{center}
% \caption{\small{Generalization of $C$ to
%     $C_{bad}$. Red circles represent regions containing only invalid
%     mappings.}}
% \label{fig-generalization}
% \end{wrapfigure}

The outline of the algorithm is shown in
\ifdefined\techreport
Algorithm~\ref{alg-synthesis}
\else Fig.~\ref{fig-algorithm}\fi. The function $synthesize$ takes four
inputs: a pair of processes, $P$ and $Q$, a desired specification $\phi$,
and a user-specified constraint $C_0$. It also stores a set of
constraints, $X$, which keeps track of ``bad'' regions of the search
space that do not contain any valid mappings.

In each iteration, the algorithm selects some mapping $m$ from $C_0$ (line 3) and
checks whether it belongs to one of the constraints in $X$ (meaning,
the mapping is guaranteed to result in an invalid implementation). If
so, it is simply discarded (lines 4-5).

% Given a candidate $C$, 
Otherwise, the verifier is used to check whether $m$ is valid with
respect to $\phi$ (line 7). If so, then $generalize$ is invoked 
to produce a maximal mapping constraint $C_{maximal}$, which represents the
largest set that contains $\{m\}$ and is valid with respect to $\phi$
(line 9). If, on the other hand, $m$ is invalid (i.e.,\ it 
fails to preserve $\phi$), then $generalize$ is invoked
to compute the largest superset $C_{bad}$ of $\{m\}$ that contains only
invalid mappings (i.e., those that satisfy $\neg \phi$).  The set
$C_{bad}$ is then added to $X$ and used to prune out subsequent,
invalid candidates (line 13).

\SetKwProg{fun}{fun}{}{end}

\ifdefined\techreport
\begin{algorithm}[!t]
\DontPrintSemicolon
  \fun{synthesize(P, Q, $\phi$, $C_0$)}{
    $X = \{ \}$\;
    \For{$m \in C_0$}{
    \If{$\exists C_{bad} \in X : m \in C_{bad}$}{
      \bf{skip}
     }
    $result \gets verify(\{m\},P,Q,\phi)$\;
    \eIf{$result$ = $OK$}{
      $C_{maximal} \gets generalize(\{m\},P,Q,\phi,C_0)$\;
      \Return{$C_{maximal}$}\;
     } {
       $C_{bad} \gets generalize(\{m\},P,Q,\neg \phi,C_0)$\;
       $X \gets X \cup \{ C_{bad} \}$\;
     }
     }
     \Return{\textnormal{none}}\;
   }
\fun{generalize(C, P, Q, $\phi$, $C_0$)}{
  $K \gets decompose(C)$\;
  \For{k $\in K$}{
    $C' \gets relax(C, k)$\;
    $result \gets verify(C',P,Q,\phi)$\;
    \If{$result$ = $OK \land C' \subseteq C_0$ }{
      $C \gets C'$
    } 
    }
    \Return{$C$}
}
 \caption{An algorithm for synthesizing a maximal mapping constraint.}
 \label{alg-synthesis}
\end{algorithm}
\else
\begin{figure}[t]
\centering
\includegraphics[width=0.85\textwidth]{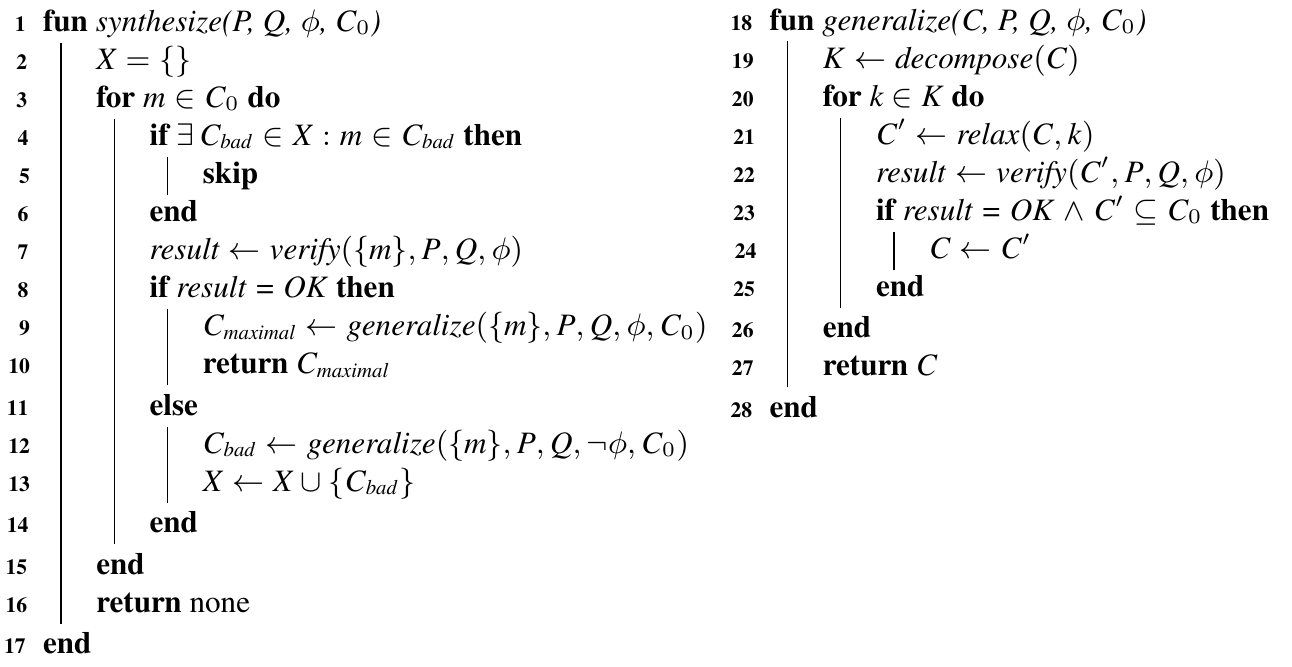}
\vspace{-2.5mm}
\caption{\small{An algorithm for synthesizing a maximal mapping constraint.}}
\label{fig-algorithm}
\vspace{-5.0mm}
\end{figure}
\fi

\para{Constraint generalization} The function $generalize(C,P,Q,\phi,C_0)$
computes the largest set that contains $C$ and only permits mappings
that satisfy $\phi$. This function is used both to identify an undesirable
region of the candidate space that should be avoided, and also to produce
a maximal version of a valid mapping constraint.

The procedure works by incrementally growing $C$ into a larger set
$C'$ and stopping when $C'$ contains at least one mapping that
violates $\phi$. Suppose that constraint $C$ is represented by a symbolic
expression $\mathcal{X}$, which itself is a conjunction of $n$ subexpressions
$k_1 \land k_2 \land ... \land k_n$, where each $k_i$ for
$1 \leq i \leq n$ represents a (possibly conditional) assignment of a
variable or a constant to some label parameter. The function
$decompose(C)$ takes the given constraint and returns the set of such
subexpressions\footnote{We assume that a symbolic expression
  representing constraint $C$ is available to $decompose$.}. The function $relax(C,k_i)$ then computes a new
constraint by removing $k$ from $C$; this new constraint, $C'$, is a
larger set of mappings that subsumes $C$.

The verifier is then used to check $C'$ against $\phi$ (line 22). If $C'$
is still valid with respect to $\phi$, it implies that the implementation
decision encoded by $k$ is irrelevant to $\phi$, meaning we can safely
remove $k$ from the final synthesized constraint $C$ (line 24). If
not, $k$ is retained as part of $C$, and the algorithm moves onto the
next subexpression $k$ as a candidate for removal (line 20).

\ifdefined\arxiv \begin{exampleregion} \fi \textbf{Example.} Back to our running example, one candidate
  constraint $C$ (line 3) is represented in part by the following expression
  (for mappings from \sf{writeEve} to \sf{encWrite}):
\begin{align*}
\mathcal{X} (a, b) \equiv\  a.m = b.m \land (b.m = \secret
                                   \implies b.k = \keyX) \land (b.m =
                                   \public \implies b.k = \keyY)
\end{align*}
The verifier returns $OK$ after checking $C$ against $P$, $Q$, and $\phi$
(line 7), meaning $C$ is a valid constraint. Next, the
generalization procedure removes the subexpression
$k_1 = (b.m = \public \implies b.k = \keyY)$ from $C$, resulting in
constraint $C'$ that is represented as:
\begin{align*}
\mathcal{X} (a, b) \equiv\  a.m = b.m \land & (b.m = \secret
                                   \implies b.k = \keyX) 
\end{align*}
When checked by the verifier (line 22), $C'$ is still considered
valid, meaning that the decision encoded by $k_1$ is irrelevant to the
property; thus, $k_1$ can be safely
removed.

However, removing $k_2 = (b.m = \secret \implies b.k = \keyX)$ results
in a violation of the property. Thus, $k_2$ is kept as part of the final
maximal constraint expression.
\ifdefined\arxiv \end{exampleregion} \fi

% \paragraph{Type-based pruning} We leverage information about types of
% variables and constants to eliminate candidate constraints before
% invoking the verifier.  More formally, let $Typ$ be a set of types,
% and let $typ : (Var \union Const) \rightarrow Typ$ the typing function
% that maps each variable or constant to its type. Given a candidate
% constraint, type-checking the expression is straightforward. Since
% each constraint can be decomposed into a set of equality expressions
% between two terms, we search for an expression that equates two terms
% with different types (i.e., $t_1 = t_2$ where
% $typ(t_1) \neq typ(t_2)$). If such expression exists, then the whole
% constraint is considered ill-formed and discarded.

% \paragraph{Learning syntactic properties of invalid mappings.} A candidate
% solution is considered to be invalid if it fails to satisfy a
% user-specified liveness or safety property. Our algorithm leverages
% information from such invalid candidates to further prune out other
% candidates that may be invalid for a similar reason. \todo{More about
%   this later}

%% file: case-study.tex
\section{Implementation and Case Study}
\label{sec-case-study}

This section describes an implementation of our synthesis
technique, and a case study on security
protocols. The goal of the study is to answer (1)
whether our technique can be used to synthesize valid implementation
mappings for realistic systems, and (2) how effective our
generalization-based algorithm is over the naive approach.

% \subsection{Public Channel}

% We successfully applied our synthesizer to generate a symbolic mapping
% constraint that ensures the validity of a resulting implementation. The
% synthesis procedure tested \todo{XX} candidates, and took
% approximately \todo{XX} seconds.

\ifdefined\arxiv
\subsection{Implementation}
\label{sec-implementation}
\else
\para{Implementation}
\fi

\ifdefined\tacas
We have built a prototype implementation of the synthesis algorithm
from Section~\ref{sec-synthesis-technique}. In particular, our tool uses the
Alloy Analyzer~\cite{alloy} as the underlying modeling and
verification engine. We chose Alloy because (1) its flexible,
declarative logic is suitable for encoding the semantics of the
mapping composition as well as specifying mapping constraints, and (2)
its SAT-based engine provides automatic property checking. However,
our approach does not prescribe the use of a particular verification
engine, and other tools may well be suitable. Our prototype is
available at \url{https://github.com/eskang/MappingSynthesisTool}.
\else 
We have built a prototype tool that is capable of performing both the
naive and generalization algorithms described in
Section~\ref{sec-synthesis-technique}. Our tool is currently plugged in with two
different verifiers: Spin~\cite{spin}, an explicit-state model
checker, and the Alloy Analyzer~\cite{alloy}, a modeling and analysis tool
based on a first-order relational logic. Both tools have their
strengths and weaknesses~\cite{zave-comparison-alloy-spin}.  Spin
provides stronger guarantees in that it is capable of fully exploring
the state space, whereas Alloy performs a kind of bounded model
checking (where the maximum length of traces explored is restricted to
some fixed bound). On the other hand, in our experience so far, Alloy
tends to find counterexamples more quickly (if they exist within the
given bound), in part thanks to its constraint solving backend.

We began with an implementation that employed Spin as a
verifier. Through our experiments, we discovered that as the size of
the relaxed constraint $C'$ increased during the iterative
generalization procedure (line 21 in Algorithm~\ref{fig-algorithm}),
the verification task by Spin became more demanding, eventually
becoming a major bottleneck in the synthesis algorithm. We then
employed the Alloy Analyzer as the verifier, and found that it did not
suffer from the same issue. We believe that this is partly due to the
constraint-based nature of Alloy: The relaxation step involves simply
removing a constraint from the Alloy model, and does not adversely
affect the performance of the constraint solver (in some cases, it
leads to improvement).
\fi

% Our current synthesizer combines the strengths of both tools: After
% synthesizing a valid mapping constraint with the Alloy Analyzer as
% the verifier, the correctness of the solution can be confirmed by
% re-verifying it with Spin (for full state exploration).

% For our current prototype, we implemented a synthesizer that
% enumerates candidate mapping constraints and employs Spin~\cite{spin}
% as a verifier during each iteration. Promela, the input language of
% Spin, expects operational models of processes, and cannot be used to
% directly encode the semantics of the mapping-based composition as
% described in Section~\ref{sec-modeling-framework}. Instead, we devised
% a scheme for encoding a mapping as an intermediate process that sits
% between a pair of high-level and low-level processes, and controls how
% their messages are synchronized. Due to limited space, we omit the a
% detailed discussion of our Spin encoding; the prototype implementation
% and the Spin models used in the experiments can be found in the
% supplementary materials.

\ifdefined\arxiv
\subsection{OAuth Protocols}
\else
\para{OAuth protocols}
\fi

As a major case study, we took on the problem of synthesizing valid
mappings for \textit{OAuth}, a popular family of
protocols used to carry out a process called \textit{third-party
  authorization}~\cite{oauth2}. The purpose of OAuth is to allow an
application (called a \textit{client} in the OAuth terminology) to
access a resource from another application (an \textit{authorization
  server}) without needing the credentials of the resource owner (an
\textit{user}). For example, a gaming application may initiate an
OAuth process to obtain a list of friends from a particular user's
Facebook account, provided that the user has authorized Facebook to
release this resource to the client.

% When building a web-based implementation of OAuth, developers are
% given freedom on how to implement abstract protocol operations as HTTP
% requests. Due to the high complexity of modern web browsers, however,
% it is challenging to reason about potential security consequences of
% implementation decisions, even for those with security expertise. As a
% result, many web-based implementations of OAuth have been shown to be
% vulnerable to attacks, some of which exploit low-level details in the
% underlying HTTP platform that do not appear in the abstract protocol
% specification~\cite{oauth-study-sun,oauth-study-wang,
%   oauth-study-chen}.

In particular, we chose to study two versions of OAuth---OAuth 1.0 and
2.0. Although OAuth 2.0 is intended to be a replacement for OAuth 1.0,
there has been much contention within the developer community about
whether it actually improves over its predecessor in terms of
security. \ifdefined\arxiv For this reason, certain major websites (such as Twitter and
Flickr) still rely on OAuth 1.0, while others have adopted 2.0. In
fact, the original creator of OAuth himself has recommended 1.0 as the
more secure version~\cite{oauth-versions}:
\begin{quote}
  ...OAuth 2.0 at the hand of a developer with deep
  understanding of web security will likely result is a secure
  implementation. However, at the hands of most developers...2.0 is likely to
  produce insecure implementations.
\end{quote}\fi
Since both protocols are designed to provide the same security
guarantees (i.e., both share common properties), our goal was
to apply our synthesis approach to systematically compare what
developers would be required to do in order to construct secure
web-based implementations of the two. \ifdefined\arxivWe began by constructing and
verifying abstract protocol models against $\phi$ to confirm that they
indeed provide the necessary security guarantees at the abstract
level. We then applied our technique to synthesize valid mappings from
the two protocols to a model of the HTTP
platform. The rest of this
section describes our experimental procedure and results.\fi

\begin{figure}[!t]
\centering
\includegraphics[width=0.93\textwidth]{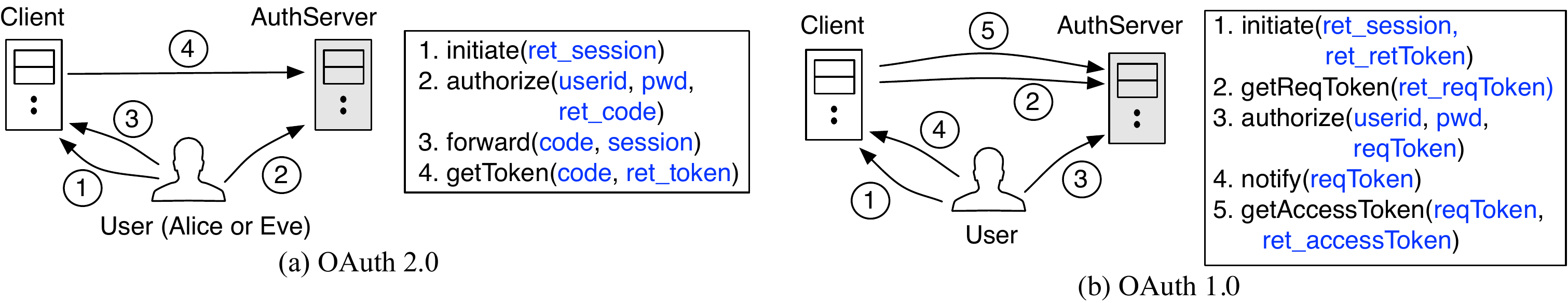}
\vspace{-2.5mm}
\caption{\small{A high-level overview of the two OAuth protocols, with a
  sequence of event labels that describe protocol steps in the typical
  order that they occur. Each arrowed edge indicates the direction of
  the communication. Variables inside labels with the prefix
  \sf{ret\_} represent return parameters. For example, in Step 2 of
  OAuth 2.0, \sf{User} passes her user ID and password as arguments
  to \sf{AuthServer}, which returns \sf{ret}\_\sf{code} back to
  \sf{User} in response.}}
\label{fig-oauth}
%\vspace{-5.0mm}
\end{figure}

\subsection{Experimental Setup}

\para{OAuth models} We constructed Alloy models of OAuth 1.0 and 2.0
based on the official protocol
specifications~\cite{oauth1,oauth2}. Due to limited space, we give
only a brief overview of the models. Each model consists of four
processes: \Client, \AuthServer, and two users, \Alice and \Eve
(latter with a malicious intent to access Alice's resources).

A typical 2.0 workflow, shown in Figure~\ref{fig-oauth}(a),
begins with a user (\Alice or \Eve) initiating a new protocol session
with \Client (\initiate). The user is then asked to prove her own
identity to \AuthServer (by providing a user ID and a password) and
officially authorize the client to access her resources
(\authorize). Given the user's authorization, the server then
allocates a unique code for the user, and then redirects her back to
the client. The user forwards the code to the client (\forward), which
then can exchange the code for an access token to her resources
(\getToken).

Like in OAuth 2.0, a typical workflow in 1.0 (depicted in
Figure~\ref{fig-oauth}(b)) begins with a user initiating a new session
with \Client (\initiate). Instead of immediately directing the user to
\AuthServer, however, \Client first obtains a \textit{request token}
from \AuthServer and associates it with the current session
(\getReqToken). The user is then asked to present the same request
token to \AuthServer and authorize \Client to access her resources
(\authorize). Once notified by the user that the authorization step
has taken place (\notify), \Client exchanges the request token for an
access token that can be used subsequently to access her resources
(\getAccessToken).

\ifdefined\tacas
There are two desirable properties of OAuth: (1) \textbf{Progress}:
Each protocol session can be completed with the client obtaining an
access token for a user, and (2) \textbf{Integrity}: When the client
receives an access token, it must correspond to the user who initiated
the current protocol session. The input formula to our algorithm,
$\phi$, is constructed as a conjunction of these two properties (due
to limited space, more details can be found in the extended
version~\cite{mapping-arxiv}).
\fi

\ifdefined\arxiv
We specified two desirable properties of OAuth: 
\begin{itemize}
\item \textbf{Progress}:
Each protocol session can be completed with the client obtaining an
access token for a user:
\begin{align*}
  \forall t \in T,&\ r \in \sf{Session}, u \in \sf{User} :
  t_0 = \initiate(r) \land t_0 \in out(u) \implies \\
&\exists t' \in T : t \leq t' \land \sf{token}_u \in knows(\Client, t')
\end{align*}
where $t_0$ refers to the last event in trace $t$, and $t \leq t'$
means $t$ is a prefix of $t'$. In other words, if some user $u$
initiates an OAuth session, \sf{Client} must eventually be able to
obtain the access token for $u$.
\item \textbf{Integrity}: When the client
receives an access token, it must correspond to a user who initiated
the protocol session in the past:
\begin{align*}
  \forall t \in T,&\ u \in \sf{User} : 
    \sf{token}_u \in knows(\Client, t) \implies \\
&\exists t' \in T, r \in \sf{Session} : t' \leq t \land l_0
  = \initiate(r) \land t_0 \in out(u) 
\end{align*}
In other words, if \Client is able to obtain a token for some user
$u$, then that same user must have initiated a session in the past.
\end{itemize}
To ensure the validity of the OAuth
models before the synthesis step, we performed verification of the
properties using Alloy; both properties were
verified in several seconds.
\fi

\para{HTTP platform model} In this case study, our goal was to explore
and synthesize \textit{web-based} implementations of OAuth. For this
purpose, we constructed a model depicting interaction between a
generic HTTP server and web browser. To ensure the fidelity of our
model, we studied, as references, similar efforts by other
researchers in building reusable models of the web for
security analysis \cite{stanford-model,bansal-webspi,trier-web-model}
(none of these models, however, has been used for synthesis).

The model contains two types of processes, \Server and \Browser (which
may be instantiated into multiple processes representing different
servers and browsers). They interact with each other over HTTP
requests, which share the following signature:
\begin{align*}
\req(method: \Method, url: \URLL, headers: \List[\Header], body: \Body, ret\_resp: \Response) 
\end{align*}
The parameters of an HTTP request have their own internal
structures, each consisting of its own parameters as follows:
\begin{align*}
&\sf{url}(host: \Host, path : \Path, queries: \List[\Query]) \quad \header(name: \Name, val: \Value) \\
&\resp(status: \Status, headers: \List[\Header], body: \Body) 
\end{align*}
Our model describes \textit{generic}, \textit{application-independent}
HTTP interactions. In particular, each \Browser process is a machine
that constructs, at each communication step with \Server, an arbitrary
HTTP request by non-deterministically selecting a value for each
parameter of the request. The processes, however, follow a
\textit{platform-specific} logic; for instance, when given a response
from \Server that instructs a browser cookie to be stored at a
particular URL, \Browser will include this cookie along with every
subsequent request directed at that URL. In addition, the model
includes a process that depicts the behavior of a web attacker, who
may operate her own malicious server and exploit weaknesses in a
browser to manipulate the user into sending certain HTTP requests.

% We spent a considerable amount of effort (around 4 months) to ensure
% that this model is as generic and reusable as possible. While building
% a faithful domain model is a challenging task, we believe that it is
% worth the effort, since it can be used to synthesize mappings for
% multiple web-based applications or protocols, thus amortizing the cost
% over time. Recently, there have been similar efforts by other
% researchers in building formal, reusable models of the web for
% security analysis \cite{stanford-model,bansal-webspi,trier-web-model}
% (although none of them has been used for synthesis tasks). We have
% studied these previous efforts as references to ensure the fidelity of
% our own models.

% \begin{align*}
% \langle  \ts{Direct(clientID,redir)} \then
%   \ts{Authorize(userID,pwd,clientID,redir)} \then 
%   \ts{Return(code,redir)} \then 
%   \ts{Forward(code,redir)}
% \rangle
% \end{align*}

\begin{figure}[!t]
\centering
\includegraphics[width=0.80\textwidth]{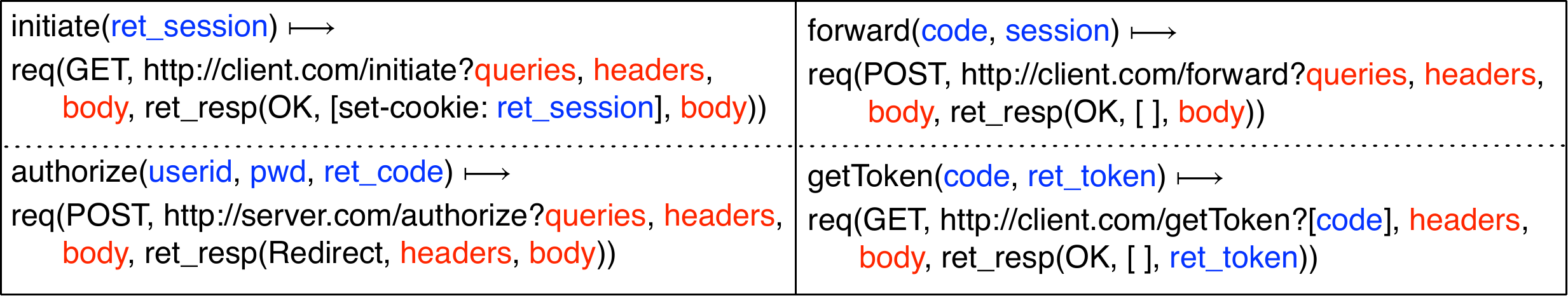}
\vspace{-2.5mm}
\caption{\small{Partial mapping specification from OAuth 2.0 to HTTP. Terms highlighted in blue and red are variables that
    represent the parameters inside OAuth and HTTP labels,
    respectively. For example, in \sf{forward}, the abstract
    parameters \sf{code} and \sf{session} may be transmitted as part
    of an URL query, a header, or the request body, although its URL
    is fixed to \sf{http://client.com/forward}.}}
\label{fig-oauth-mapping}
%\vspace{-5mm}
\end{figure}

\para{Mapping from OAuth to HTTP} Building a web-based implementation
of OAuth involves decisions about how abstract protocol operations are
to be realized in terms of HTTP requests. As an input to the
synthesizer, we specified an initial set of constraints that describe
partial implementation decisions; the ones for OAuth 2.0 are shown in
Figure~\ref{fig-oauth-mapping}. These decisions include a designation
of fixed host and path names inside URLs for various OAuth operations
(e.g., \sf{http:/client.com/initiate} for the OAuth \sf{initiate}
event), and how certain parameters are transmitted as part of an HTTP
request (\sf{ret}\_\sf{session} as a return cookie in
\sf{initiate}). It is reasonable to treat these
constraints as given, since they describe decisions that are common
across typical web-based OAuth implementations.

\ifdefined\arxiv
\para{Associativity}

Typically, in our methodology, the designer would start by composing
the processes in an abstract protocol model (i.e., OAuth) to analyze
the desired properties of the protocol before mapping it into the
underlying platform; then, the process that represents the overall
system, $Sys_1$, would be constructed as follows:
\begin{align*}
& Sys_1 = (\mparw{\Alice}{\Browser}{m_1}) \parallel
  (\mparw{\AuthServer}{\Server_{AS}}{m_2}) \parallel
  (\mparw{\Client}{\Server_{C}}{m_3}) \\
& m_1 = \{ (\sf{initiate}, \sf{req}_i), (\sf{authorize}, \sf{req}_a),
  (\sf{forward}, \sf{req}_f)\} \\
& m_2 = \{ (\sf{authorize}, \sf{req}_a), (\sf{getToken}, \sf{req}_g) \} \\
& m_3 = \{ (\sf{initiate}, \sf{req}_i), (\sf{forward}, \sf{req}_f),
  (\sf{getToken}, \sf{req}_g) \}
\end{align*}
Alternatively, the designer may first decide how each abstract process
is implemented on the platform, and then compose the resulting
concrete processes to construct the overall system, $Sys_2$, as follows:
\begin{align*}
& Sys_2 = \mparw{(\Alice \parallel \AuthServer \parallel
  \Client)}{(\Server_{AS} \parallel \Server_{C} \parallel
  \Browser)}{m} \\
& m = m_1 \cup m_2 \cup m_3 = \{ (\sf{initiate}, \sf{req}_i), (\sf{authorize}, \sf{req}_a), (\sf{forward}, \sf{req}_f), (\sf{getToken}, \sf{req}_g)\}
\end{align*}
The former approach is useful in that it enables \emph{incremental
analysis} (i.e., reason about the properties of an abstract design before
verifying the more detailed concrete implementation), while the latter
facilitates \emph{incremental implementation} (e.g.,\ implement \Client
first before implementing \AuthServer). In this case
study, note that the processes and mappings satisfy the condition for
the associativity theorem (Theorem~\ref{theorem-mapping-associativity}):
\begin{align*}
events(Sys_1) = events(Sys_2) = \{ & \e{\sf{initiate}, \sf{req}_i},
  \e{\sf{authorize}, \sf{req}_a}, \e{\sf{forward}, \sf{req}_f}, \\
  & \e{\sf{getToken}, \sf{req}_g}, \e{ \sf{req}_i},  \e{ \sf{req}_a}, \e{ \sf{req}_g}\}
\end{align*}
Thus, both approaches to mapping yield an equivalent OAuth
implementation model.
\fi

\subsection{Results} 

Our synthesis tool was able to generate valid mapping
constraints for OAuth 2.0 and 1.0. In particular, the constraints
describe mitigations against attacks that exploit an interaction
between the OAuth logic and a browser vulnerability, including
\textit{session swapping}~\cite{oauth-study-sun}, \textit{covert
  redirect}~\cite{oauth-advisory2} (both for OAuth 2.0), and
\textit{session fixation}~\cite{oauth-advisory1} (for OAuth 1.0). We
describe an example of how certain platform decisions may result in an
attack, and how our synthesized mappings mitigate against this attack.

\para{Insecure mapping} Consider OAuth 2.0 from Figure~\ref{fig-oauth}(a).  In order to
implement the \sf{forward} operation, for instance, the developer must
determine how the parameters $code$ and
$session$ of the abstract event label 
are encoded using their concrete counterparts in an HTTP
request. A number of choices is available. In one possible
implementation, the authorization code may be transmitted as a query
parameter inside the URL, and the session as a browser cookie, as
described by the following constraint expression, $\mathcal{X}_1$:
\begin{align*}
  \mathcal{X}_1(a,& b) \equiv 
(b.method = \sf{POST}) \land (b.url.host = \sf{client.com}) \land \\
&  (b.url.path = \sf{forward}) \land (b.url.queries[0]  = a.code) \land \\
&  b.headers[0].name = \sf{cookie} \land  b.headers[0].value = a.session
 \end{align*}
where \sf{POST}, \sf{client.com}, \sf{forward}, and \sf{cookie} 
are predefined constants; and $l[i]$ refers to $i$-th element
of list $l$.

This constraint, however, allows a vulnerable implementation 
% that fails to satisfy the integrity property for OAuth. In this attack,
where
malicious user \Eve performs the first two steps of the workflow in
Figure~\ref{fig-oauth}(a) using her own credentials, and obtains a
unique code ($\sf{code}_{\mathsf{Eve}}$) from the authorization
server. Instead of forwarding this to \sf{Client} (as she is expected
to), Eve keeps the code herself, and crafts her own web page that
triggers the visiting browser to send the following HTTP request:
\begin{align*}
\req(\sf{POST},
  \sf{http://client.com/forward?code}_{\mathsf{Eve}}, ...) 
\end{align*}
Suppose that Alice is a naive browser user who may occasionally be
enticed or tricked into visiting malicious web sites. When Alice
visits the page set up by Eve, Alice's browser automatically generates
the above HTTP request, which, given the decisions in
$\mathcal{X}_1$, corresponds to a valid \sf{forward} event:
\begin{align*}
& \forward(\sf{code}_{\mathsf{Eve}}, \sf{session}_{\mathsf{Alice}}) \mapsto \\
& \req(\sf{POST}, \sf{http://client.com/forward?code}_{\mathsf{Eve}},
[(\sf{cookie}, \sf{session}_{\mathsf{Alice}})], ...) 
\end{align*}
Due to the standard browser logic, the cookie corresponding to
$\sf{session}_{\mathsf{Alice}}$ is included in every request to
\sf{client.com}. As a result, \Client mistakenly accepts
$\sf{code}_\mathsf{Eve}$ as the one for Alice, even though it 
belongs to Eve, violating the integrity property of OAuth. 

\para{Synthesized, secure mapping} 
% A major contributing factor to the above attack is 
Mapping expression $\mathcal{X}_1$ is incorrect as it allows insecure
mappings. Our tool is able to automatically synthesize a secure
mapping expression $\mathcal{X}_2$, described below. $\mathcal{X}_2$
fixes the major problem of $\mathcal{X}_1$, namely, that in a
browser-based implementation, the client cannot trust an authorization
code as having originated from a particular user (e.g.,\ \sf{Alice}),
since the code may be intercepted or interjected by an attacker
(\sf{Eve}) while in transit through a browser. A possible solution is
to explicitly identify the origin of the code by requiring an
additional piece of tracking information to be provided in each
\sf{forward} request. The mapping expression $\mathcal{X}_2$
synthesized by our tool, given the input partial constraint in
Figure~\ref{fig-oauth-mapping}, encodes one form of this solution:
\begin{align*}
  \mathcal{X}_2&(a, b) \equiv \mathcal{X}_1(a, b)\land 
% (b.method = \sf{POST}) \land (b.url.host = \sf{client.com}) \land \\
% &  (b.url.path = \sf{forward}) \land (b.url.queries[0]  = a.code) \land \\
% &  b.headers[0].name = \sf{cookie} \land  b.headers[0].value = a.session
% \end{align*}
% \begin{align*}
 (a.session = \sf{session}_{\mathsf{Alice}} \implies b.url.queries[1] = \sf{nonce}_0) \land\\
&  (a.session = \sf{session}_{\mathsf{Eve}} \implies b.url.queries[1] = \sf{nonce}_1)  
 \end{align*}
 where $\sf{nonce}_o$, $\sf{nonce}_1 \in \sf{Nonce}$ are constants
 defined in the HTTP model\footnote{A nonce is a unique piece of
   string intended to be used once in communication.}. In particular,
 $\mathcal{X}_2$ stipulates that every \sf{forward} request must include an
 additional value (\sf{nonce}) as an argument besides the code and 
 the session, and that this nonce be unique for each session value. 
 $\mathcal{X}_2$ ensures that the resulting implementation
 satisfies the desired properties of OAuth 2.

 % Although the above attack and possible mitigations have been
 % informally discussed within the security community (referred as
 % session swapping in~\cite{oauth-study-sun}), our study is the first
 % to automatically synthesize implementation decisions that 
 % mitigate against the attack. 

 \ifdefined\arxiv
\para{OAuth 1.0 vs 2.0} Based on the comparison of the two protocol
workflows in Figure~\ref{fig-oauth}, OAuth 2.0 appears simpler than
its predecessor, which requires the client to perform an extra step to
obtain a request token from the authorization server (Step 2 in
Figure~\ref{fig-oauth}(b)).

Simplicity, however, sometimes can result in a loss of security. In
OAuth 1.0, the client knows exactly the request token that it expects
to receive in Step 4 (the same one as in Step 2). Thus, it needs not
trust that the user will always deliver a correct request token in
Step 4, and does not suffer from the above attack.  OAuth 2.0, on the
other hand, relies on the user to deliver a correct authorization code
(Step 3 in Figure~\ref{fig-oauth}(a))---an assumption which does not
necessarily hold in a browser-based implementation, as described
above. In effect, the synthesized mapping $\mathcal{X}_2$ captures a way to
harden OAuth 2.0 implementations against security risks of relying on
this assumption.
\fi

% It turns out, however, that OAuth 1.0 implementations may be
% vulnerable to a different kind of attack (called \textit{session
%   fixation}~\cite{oauth-advisory1}). Our synthesized 
% constraint for 1.0 mitigates against this attack.

\para{Performance} Figure~\ref{fig-results} shows statistics from our
experiments synthesizing valid mapping constraints
for the two OAuth protocols\footnote{The experiments were performed on
  a Mac OS X 2.7 GHz machine with 8G RAM and MiniSat~\cite{minisat} as
  the underlying SAT solver employed by the Alloy Analyzer.}. Overall,
the synthesizer took approximately 17.6 and 24.0 minutes to synthesize
the constraints for 1.0 and 2.0, respectively. This
level of performance is acceptable for \textit{generic} protocols such as
OAuth: Once synthesized, the implementation guidelines captured by the
constraints can be consulted by \textit{multiple} developers to build
their own OAuth implementations, amortizing the cost of the
synthesis effort over time.

\begin{figure}[!t]
\centering
  \includegraphics[width=0.90\textwidth]{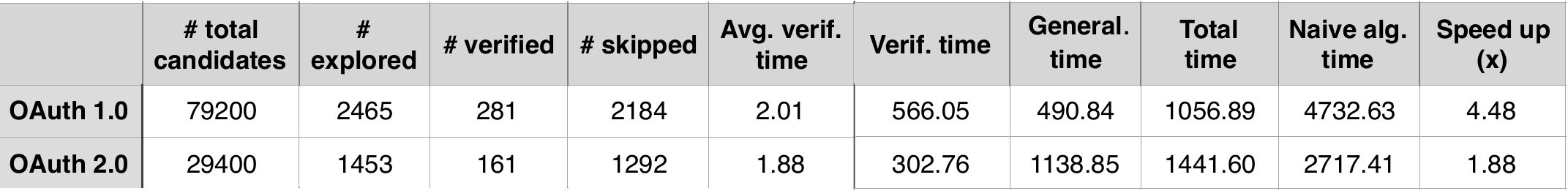}
\vspace{-2.5mm}
\caption{\small{Results from the experiments (all times in
    seconds). ``\# total candidates'' is the total number of possible
    mappings; ``\# explored'' is the number of
    iterations taken by the main synthesis loop (lines 3-15,
    Figure~\ref{fig-algorithm}) before a solution was found. ``Verif.\
    time'' and ``General.\ time'' are the total amounts of time spent
    on verification and generalization, respectively. ``Total time''
    refers to the time spent by the generalization-base algorithm to
    synthesize a maximal constraint.}}
\label{fig-results}
%\vspace{-5.0mm}
\end{figure}

In both cases, the tool spent a considerable amount of time on the
generalization step to learn the invalid regions of the search
space. It can also be seen that the generalization is effective at
identifying and discarding a large number of invalid candidates, and
achieves a significant amount of speed up over the naive algorithm
(4.48 and 1.88 times for OAuth 1.0 and 2.0, respectively). 

% One noticeable aspect of the results is that the synthesizer was able
% to find a valid solution after exploring a small part of the search
% space.\todo{More}

\ifdefined\arxiv
Since OAuth 1.0 is a more complex protocol than 2.0, we expected the
synthesis procedure to take more time on the former; indeed, the
performance of the naive algorithm was consistent with this
expectation (4732.63 vs 2717.41 seconds). However, we were surprised
to find that the generalization-based algorithm was able to find a
solution for OAuth 1.0 faster than it did for 2.0. In particular, the
synthesizer spent a significantly larger proportion of its time on the
generalization step ($1138.85/1441.60 \approx 79\%$); one possible
explanation is that the search space for OAuth 2.0 is more densely
populated with \textit{incomparable} invalid constraints than 1.0.
\fi

As an additional experiment, we ran the generalization-based algorithm
to explore the search space exhaustively without terminating when a
solution is found. The algorithm was able to skip 75375 out of
79200 candidates for OAuth 1.0 (roughly $95\%$), and 28140 out of
29400 candidates for OAuth 2.0 ($\approx 96\%$). \ifdefined\arxiv
This means that only 3825 and 1260 calls to the verifier were needed
to exhaust the search space (for 1.0 and 2.0,
respectively)---relatively small compared to the total number of
candidates.\fi\ Our generalization technique was particularly effective
for the OAuth protocols, since a significant percentage of the
candidate constraints would result in an implementation that violates
the progress property (i.e., it prevents \Alice or \Eve from
completing a protocol session in an expected order). Often, the
decisions contributing to this violation could be localized to a small
subset of entries in a mapping (for example, attempting to send a
cookie to a mismatched URL, which is inconsistent with the behavior of
the browser process). By identifying this subset, our algorithm was
able to discover and eliminate a large number of invalid mappings.

%\subsection{Discussions}

\ifdefined\techreport
\para{Lessons learned} 

The input mapping constraint (depicted in
Figure~\ref{fig-oauth-mapping}), capturing the developer's partial
knowledge or decisions, was crucial for
reducing the size of the search space and keeping the synthesis
procedure tractable. We believe that in practice, our tool will be
more effective as a \textit{completion} tool---allowing the developer
to express her partial knowledge using a constraint, and filling in
gaps that represent her uncertainty---rather than one that generates
an entire set of implementation decisions from scratch.
\fi

% Partial mapping: Crucial for reducing the search space, but sometimes
% may introduce overconstraint that prevents a valid mapping from
% synthesize. But this can be made interactive; if no mapping found, try
% different ones. 

\para{Threats to validity} \ifdefined\techreport One potential source of errors in our study
is the scope and accuracy of the models used to describe the OAuth
protocols and HTTP platform. The security guarantees provided by
synthesized mappings are limited to the extent which these models are
accurately able to capture potential attacks on real OAuth
implementations. For example, our synthesized mappings do not protect
an implementation against attacks that exploit a vulnerability called
\textit{cross-site scripting}: Modeling this vulnerability would involve
details about how individual HTML pages are sanitized, which, in turn,
depend on the underlying web development framework or libraries
used. However, we believe that this is not an inherent flaw in our
approach, but rather a risk that may arise in any
model-based approach to verification and synthesis.\fi

One potential source of errors stems from the finitization of the
system models. To finitize the set of labels $L$ and ensure that the
synthesis procedure terminates, we bounded the size of datatype
domains to 4 (i.e., 4 access tokens, cookies, etc.). While we believe
that this bound is sufficient to explore possible interactions between
protocol participants during an OAuth session, it is possible that we
might have missed a potential security violation involving a larger
number of data elements. To mitigate this risk, we plan to explore
infinite-state verification methods (e.g., using SMT solvers~\cite{z3,yices}) as part of
future work.

%  To increase the fidelity of models, we
% plan to explore techniques for automatically extracting models from
% code or system logs~\cite{song-extraction,authscan,model-inference1}.

% Another possible source of errors is the bounded nature of the
% verification techniques by the Alloy Analyzer. Since it
% explores the behavior of the system only up to certain bounds, it is
% possible that it may have missed an attack beyond those bounds. In
% general, selecting an appropriate bound requires insights about the
% problem domain. At the same time, it is also crucial for keeping the
% verification task tractable.  

%\todo{Eunsuk:\paragraph{Threats to validty.} }

%% file: related-work.tex
\section{Related Work}
\label{sec-related-work}
\ifdefined\tacas\vspace{-2.5mm}\fi

A large body of literature exists on \textit{refinement-based} methods
to system construction~\cite{hoare-refinement,back-refinement}. These
approaches involve building an implementation $Q$ that is a behavioral
refinement of $P$; such $Q$, by construction, would satisfy the
properties of $P$. In comparison, we start with an assumption that $Q$
is a \textit{given} platform, and that the developer may not have the
luxury of being able to modify or build $Q$ from scratch. Thus, instead of
behavioral refinement (which may be too challenging to achieve), we aim
to preserve some critical property $\phi$ when $P$ is implemented
using $Q$.

% Our algorithm can be considered as a kind of
% counterexample-driven inductive synthesis
% (CEGIS)~\cite{armando-sketch-cegis}, in which the verifier is used as
% an oracle to learn negative regions of the search space and guide the
% synthesizer towards a desirable solution.

The task of synthesizing a valid mapping can be seen as a type of the
\textit{model merging} problem~\cite{model-merging}. This problem has
been studied in various contexts, including architectural
views~\cite{maoz-views}, behavioral
models~\cite{merging-statecharts,merging-different-vocabularies,merging-partial-behavioral-models},
and database schemas~\cite{merging-database-schemas}. Among these, our
work is most closely related to merging of partial behavioral models
\cite{merging-different-vocabularies,merging-partial-behavioral-models}. In
these works, given a pair of models $M_1$ and $M_2$, the goal is to
construct $M'$ that is a behavioral refinement of both $M_1$ and
$M_2$. The approach proposed in this paper differs in that (1) the
mapping composition involves merging a pair of events with distinct
alphabet labels into a single event that retains all of those labels,
and (2) the composed process $(P \compo Q)$ needs not be a behavioral
refinement of $P$ or $Q$, as long as it satisfies property $\phi$.

% Full abstraction is a notion developed to reason about
% equivalence between the operational and denotational semantics of a
% programming language~\cite{milner-77,plotkin-77}. This notion has been
% adopted in the context of security, to reason about issues that arise
% when a program in one language is translated to a (typically,
% lower-level) representation in another
% language~\cite{abadi-protection}. A property violation that arises due
% to an invalid mapping can be seen as a kind of full abstraction
% violation. However, relatively little work has been done exploring
% synthesis methods for achieving full abstraction; we believe that our
% work is a step towards this goal.

Bhargavan and his colleagues presents a compiler that takes a
high-level program written using \textit{session
  types}~\cite{honda-session-types} and automatically generates a
low-level implementation~\cite{bhargavan-protocol-synthesis}. This
technique is closer to compilation than to synthesis in that it uses a
fixed translation scheme from high-level to low-level operations in a
specific language environment (.NET), without searching a space of
possible translations.

%  In comparison, our approach aims to be a
% general methodology, and can be used to synthesize mappings onto any
% target platform (as long as the model of the latter is available).

% Action refinement~\cite{action-refinement,glabbeek-action-refinement}
% is an algebraic approach used to specify and reason about how abstract
% events in a system are realized in terms of low-level events. Although
% this approach is similar to the mapping-based modeling framework, it
% has not been adapted in the context of automated synthesis.

% Specware is a deductive synthesis framework for deriving an
% implementation from a high-level specification through a series of
% refinement steps~\cite{specware}. Specware is more general in that
% their goal is to provide a general-purpose development environment,
% whereas mappings described in this paper target one kind of
% implementation decisions. On the other hand, Specware requires a
% considerable amount of manual guidance from the user during the
% refinement steps, while our goal is to enable fully automated
% synthesis.

Our approach is similar to a number of other synthesis
frameworks~\cite{oracle-guided-synthesis,rishabh-storyboard,armando-sketch,ScenariosHVC2014,CompletionCAV2015,AlurTripakisSIGACT17}
in allowing the user to provide a partial specification of the
artifact to be synthesized (in form of constraints or examples),
having the underlying engine \textit{complete} the remaining
parts. This strategy is crucial for pruning the space of candidate
solutions and keeping the synthesis problem tractable. Synthesizing a
low-level implementation from a high-level specification has also been
studied in the context of data structures~\cite{fiat-synthesis,hawkins-synthesis2,hawkins-synthesis1}.

% A similar notion of mapping is often encountered in hardware and
% embedded system design (e.g., mapping an application to an execution
% platform or
% architecture~\cite{tripakisLCTES03,SangiovanniEEtimes02,TripakisTOC08}).
% As far as we know, however, the problem of mapping synthesis has not
% been formalized within these context.

% The prior work by Kang et. al~\cite{eunsuk-security} addresses the
% problem of verifying whether a given mapping preserves a property, and
% does not provide a method for synthesizing a property-preserving
% mapping.

% Abadi notes how achieving full abstraction in
% general is difficult, since we typically do not have control over the
% underlying low-level platform.

%% file: conclusion.tex
\vspace{-10pt}

\section{Conclusions}
\label{sec-conclusions}
\ifdefined\tacas\vspace{-2.5mm}\fi

In this paper, we have proposed a novel system design methodology
centered around the notion of \textit{mappings}. We have presented
novel \emph{mapping synthesis problems} and an algorithm for
efficiently synthesizing valid mappings. In addition, we have validated our
approach on realistic case studies involving the OAuth
protocols. 

\ifdefined\arxiv
As with many synthesis problems, scalability remains a challenge. One
promising direction is to exploit the fact that our
generalization-based algorithm (from Section~\ref{sec-synthesis-technique}) is
easily parallelizable. We are currently devising an algorithm
where multiple machines are used to explore different regions of the
search space in parallel, only communicating to update
and check the invalid constraint set $X$.

Our synthesis tool can be used in an interactive manner. If it fails
to find a valid mapping constraint due to a given constraint $C$ that
is too restrictive, the developer may relax $C$ and re-run the
synthesis procedure. However, the tool currently does not provide an
explanation for \textit{why} it fails to synthesize a mapping; such an
explanation could point to parts of $C$ that must be relaxed, or
certain behavior of $Q$ that entirely precludes valid mappings. We are
developing a root cause analysis technique that can be used to provide
such explanations.

Another major next step is to bridge the gap between 
synthesized mappings (which describe implementation decisions at a
modeling level) and code. For instance, a mapping may be used to
directly generate a working implementation that preserves a desired
property (e.g.,\ a secure, reference OAuth implementation), or used as
a code-level specification to check that a program adheres to the
decisions described in the mapping. We are also exploring potential
applications of our synthesis approach to other domains
where a similar type of mapping arises, such as cyber-physical and
embedded systems~\cite{tripakisLCTES03,SangiovanniEEtimes02,TripakisTOC08}.
\else
Future directions include performance improvements (e.g.,
exploiting the fact that our generalization-based algorithm is easily
parallelizable) and application of our synthesis approach to other
domains beside security (e.g.,\ platform-based design and embedded
systems~\cite{tripakisLCTES03,sangiovanni-platform-based-design,TripakisTOC08}).
\fi

%% file: appendix-proofs.tex
\section{Proofs}

\subsection{Associativity of Mapping Composition}
\label{appen-thm-associativity}

Let us first prove several lemmas that will enable us to prove a
general theorem about associativity of mapping composition.
\begin{lemma}
\label{lem-project} 
Given trace $t \in T(L)$ and sets of labels $X, Y \subseteq L$ such
that $X \subseteq Y$, $(t \project Y) \project X = t \project X$.
\end{lemma}
\begin{proof}
  Consider event $e \in events(t)$. Since $X \subseteq Y$,
  $(e \cap Y) \cap X = e \cap X$. Thus, by the definition of the
  projection operator, every event in $(t \project Y) \project X$ is
  also an event in $t \project X$.
\end{proof}

\begin{lemma}
\label{lemma-mapping1}
Given processes $P$, $Q$, and $R$, let $X = \mparw{(\mparw{P}{Q}{m_1})}{R}{m_2}$.
Then, for any trace $t \in T(L)$, $t\in\traces(X)$ if and only if all the following conditions hold:
\begin{enumerate}
\item $t\in (events(X))^*$ 
\item $(t \project \alpha(P)) \in \traces(P)$
\item $(t \project \alpha(Q)) \in \traces(Q)$
\item $(t \project \alpha(R)) \in \traces(R)$
\end{enumerate}
\end{lemma}
\begin{proof}
($\implies$ direction) Suppose that $t \in \traces(X)$. By the definition of
$traces(\mparw{(\mparw{P}{Q}{m_1})}{R}{m_2})$, it follows that
\begin{align*}
t \in (events(X))^* \land (t \project \alpha(\mparw{P}{Q}{m_1})) \in traces(\mparw{P}{Q}{m_1}) \land (t \project \alpha(R)) \in \traces(R) 
\end{align*}
satisfying Conditions (1) and (4). 

Let $t' = t \project \alpha(\mparw{P}{Q}{m_1}) = t \project \alpha(P
\cup Q)$. By the definition of $traces(\mparw{P}{Q}{m_1})$,
\begin{align*}
& t' \in (events(\mparw{P}{Q}{m_1}))^* \land (t' \project \alpha(P)) \in
  traces(P) \land (t' \project \alpha(Q)) \in \traces(Q) \iff \\
& t' \in (events(\mparw{P}{Q}{m_1}))^* \land ((t \project \alpha(P
\cup Q)) \project \alpha(P)) \in
  traces(P) \land ((t \project \alpha(P
\cup Q)) \project \alpha(Q)) \in \traces(Q)
\end{align*}
Based on Lemma~\ref{lem-project}, we can further conclude the following:
\begin{align*}
t' \in (events(\mparw{P}{Q}{m_1}))^* \land (t \project \alpha(P)) \in
  traces(P) \land (t \project \alpha(Q)) \in \traces(Q) \tag{\bf{A}}
\end{align*}
from which it follows that Conditions (2) and (3) hold.

($\Leftarrow$ direction) Suppose that the above four conditions
hold. Since $t \in (events(X))^*$, the
following statement holds for any $e \in \events(t)$ by definition:
\begin{align}
e \cap \alpha(\mparw{P}{Q}{m_1}) \in events(\mparw{P}{Q}{m_1}) \land e
  \cap \alpha(R) \in events(R) \land cond'(e) \land condmap(e,m) \tag{\bf{B}}
\end{align}
Let $t' = t \project \alpha(\mparw{P}{Q}{m_1})$.  For any event
$e' \in events(t')$, there must exist some event $e \in events(t)$
such that $e' = e \cap \alpha(\mparw{P}{Q}{m_1})$. Thus, by the first
conjunct in \bf{(B)}, $e' \in events(\mparw{P}{Q}{m_1})$. Since $t'$
consists of only events in $events(\mparw{P}{Q}{m_1})$, it follows
that $t' \in (events(\mparw{P}{Q}{m_1}))^*$.  Along with Conditions
(2) and (3), and \bf{(A)}, we can further conclude that
\begin{align*}
t' = t \project \alpha(\mparw{P}{Q}{m_1}) \in traces(\mparw{P}{Q}{m_1})
\end{align*}
Adding Conditions (1) and (4), we derive
\begin{align*}
t\in (events(X))^* \land t \project \alpha(\mparw{P}{Q}{m_1}) \in traces(\mparw{P}{Q}{m_1}) \land (t \project \alpha(R)) \in \traces(R) 
\end{align*}
which implies $t \in traces(X)$, as required.
\end{proof}

\begin{lemma}
\label{lemma-mapping2}
Given processes $P$, $Q$, and $R$, let $Y = \mparw{P}{(\mparw{Q}{R}{m_2})}{m_1}$. 
Then,  for any trace $t \in T(L)$, $t\in\traces(Y)$ if and only if all the following conditions hold:
\begin{enumerate}
\item $t\in (events(Y))^*$ 
\item $(t \project \alpha(P)) \in \traces(P)$
\item $(t \project \alpha(Q)) \in \traces(Q)$
\item $(t \project \alpha(R)) \in \traces(R)$
\end{enumerate}
\end{lemma}
\begin{proof}
The proof is symmetric to that of Lemma~\ref{lemma-mapping1} and thus
is omitted.
\end{proof}

\begin{theorem}
\label{theorem-mapping-associativity}
Given processes $P$, $Q$, and $R$, let $X = \mparw{(\mparw{P}{Q}{m_1})}{R}{m_2}$ and $Y = \mparw{P}{(\mparw{Q}{R}{m_4})}{m_3}$. If $\events(X) = \events(Y)$,
then $X=Y$.
\end{theorem}
\begin{proof}
  Note that
  $\labels(X)= \labels(P) \cup \labels(Q) \cup \labels(R)
  = \labels(Y)$.  Along with the assumption that
  $\events(X)=\events(Y)$, it remains to show that
  $\traces(X)=\traces(Y)$. Consider an arbitrary trace $t \in
  \traces(X)$. By Lemma~\ref{lemma-mapping1}, the following four
  conditions hold:
\begin{enumerate}
\item $t\in (events(X))^*$ 
\item $(t \project \alpha(P)) \in \traces(P)$
\item $(t \project \alpha(Q)) \in \traces(Q)$
\item $(t \project \alpha(R)) \in \traces(R)$
\end{enumerate}
Since $\events(X) = \events(Y)$, we can apply these conditions to
Lemma~\ref{lemma-mapping2} to derive that $t \in (events(Y))^*$. Thus,
$\traces(X) = \traces(Y)$ and consequently, $X = Y$.
\end{proof}

From Theorem~\ref{theorem-mapping-associativity}, we can derive a more
restricted version of the condition under which the composition
operator is associative: 
\begin{corollary}
\label{corollary-mapping-associativity}
Given processes $P$, $Q$, and $R$, let $X = \mparw{(\mparw{P}{Q}{m_1})}{R}{m_2}$ and $Y = \mparw{P}{(\mparw{Q}{R}{m_2})}{m_1}$. If $\events(X) = \events(Y)$,
then $X=Y$.
\end{corollary}
\begin{proof}
It follows as a special case of
Theorem~\ref{theorem-mapping-associativity} with $m_3 = m_1$ and $m_4
= m_2$. 
\end{proof}